\documentclass{article}

\RequirePackage{amsthm,amsmath,amsfonts,amssymb}
\RequirePackage[authoryear,round]{natbib}
\RequirePackage[colorlinks,citecolor=blue,urlcolor=blue]{hyperref}
\RequirePackage{graphicx}

\usepackage{arxiv}
\usepackage{3_preamble}

\title{Causal Inference Using Augmented Epidemic Models}
\renewcommand{\shorttitle}{Causal Inference Using Augmented Epidemic Models}

\author{ 
    Heejong Bong \\
    Department of Statistics\\
    Purdue University\\
    West Lafayette, IN 47907 \\
    \texttt{bong0@purdue.edu} \\
    \And
    Valérie Ventura \qquad Larry Wasserman \\
    Department of Statistics \& Data Science \\
    Carnegie Mellon University\\
    Pittsburgh, PA 15213 \\
    \texttt{\{vventura, larry\}@stat.cmu.edu} \\
    \AND
}

\hypersetup{
pdftitle={Causal Inference Using Augmented Epidemic Models},
pdfsubject={},
pdfauthor={Heejong Bong, Valérie Ventura, Larry Wasserman},
pdfkeywords={Null paradox, Causal inference, Epidemic models, Inverse propensity weighting, Estimating equations},
}

\begin{document}


\maketitle

\begin{abstract}
Epidemic models describe the evolution of a communicable disease over time. These models are often modified to include the effects of interventions (control measures) such as vaccination, social distancing, school closings etc. Many such models were proposed during the COVID-19 epidemic. Inevitably these models are used to answer the question: What is the effect of the intervention on the epidemic? These models can either be interpreted as data generating models describing observed random variables or as causal models for counterfactual random variables. These two interpretations are often conflated in the literature. We discuss the difference between these two types of models, and then we discuss how to estimate the parameters of the model. Our focus is causal inference for parameters in epidemic models by adjusting for confounders, allowing time varying interventions.
\end{abstract}



\section{Introduction} \label{sec::intro}

In this paper we consider the problem of inferring
the causal effects of time-varying interventions $A_t$ in epidemics.
The term {\em intervention} can refer to
control measures, treatments, public health policies or 
spontaneous changes in population behavior such as reduced mobility. 
Such interventions often
change over time, often depending on the state of the epidemic.
For example, we may want to estimate the effect
of vaccinations, masks or social mobility
on the number of infections
or hospitalizations $Y_t$ over time.

{We consider the situation when only one epidemic time series $Y_t$ has been observed,
for example, Covid deaths in a single US state.
In this case, model free causal effects, such as the difference of mean 
effects between treated and non-treated groups, cannot be used
so some sort of modeling assumption is required.
A common approach in causal inference is to assume 
a {\em marginal structural model (MSM, \citealp{robins2000marginal})}, 
which is a type of causal model for time varying
inference that describe how 
the intervention affects the outcome.
Furthermore, the assumed MSM is typically simple and easily interpretable.
For 
example, letting 
$Y_t(a_1,\ldots, a_t)$ 
denote the outcome 
that would be obtained at time $t$
if the intervention variables
$A_1,\ldots, A_t$ were set to
$a_1,\ldots, a_t$,
we might use the model
\begin{equation}
\label{eq::MSM0}
\E[Y_t(a_1,\ldots, a_t)] = \beta_0 + \beta_A \sum_{j=1}^{t} a_j,
\end{equation}
which says that the expected counterfactual outcome 
$Y_t(a_1,\ldots, a_t)$ at time 
$t$ is a linear function of 
the cumulative dose $\sum_{s=1}^t a_s$ up to $t$.}
An interesting question is: if an epidemic model 
is used to describe the data generating process,
how can we use it to guide the choice of a MSM?
    
{
An example of an epidemic model is the semi-mechanistic model of 
\citet{bhatt}:
\begin{equation} 
\label{eq::bhatt}
\begin{aligned}
\Exp[I_t | \bar{A}_t, \bar{I}_{t-1}, \bar{Y}_{t-1}] &= R_t \sum_{s<t}  g_{t-s} I_s, \\
\Exp[Y_t | \bar{A}_t, \bar{I}_t, \bar{Y}_{t-1}] &= \alpha_t \sum_{s<t} \pi_{t-s} I_s,
\end{aligned}
\end{equation}
where 
$I_t$ are the unobserved infections at $t$, $Y_t$ are the observed hospitalized, deaths or cases,
$\alpha_t$ and $R_t$ are the ascertainment rate and reproduction number,
$\pi$ is the infection to death distribution and
$g$ is the generating distribution.
To study the effect of an intervention $A_t$ on the epidemic,
\citet{bhatt} modeled $R_t$ as
\begin{equation} 
\label{eq::bhatt_R}
R_t \equiv R(\bar{A}_t, \beta) = \frac{K}{1 + \exp(\beta_0 + \beta_A A_t)},
\end{equation}
where $K$ is the maximum transmission rate.
We refer to the model in \cref{eq::bhatt,eq::bhatt_R} as {\em augmented}, since it is extended to incorporate the causal effect of $A_t$.
Other augmented models that have been used in the literature are described in Section~\ref{sec::epi}.
}

{We then face the following questions: how should we interpret
an augmented model? How should we estimate the parameters of the model?
How should we use the model?}
Augmented epidemic models have been used both as data generating models (DGMs)
and as MSMs. 
{}{For example, \cite{bhatt} and \cite{bong2024} produce scenario predictions from the 
model in \cref{eq::bhatt,eq::bhatt_R} -- thereby treating it like
an MSM --
after they fit it 
to data using the likelihood function in a Bayesian or
frequentist framework, respectively -- thereby treating it like a DGM. }
But DGMs and MSMs are, in general, not the same.
For example,
suppose we observe data
$(A_t,Y_t)$ at two time points $t=1$ and $t=2$.
In the data generating interpretation, the model characterizes the conditional density of the outcome $Y_2$, which by the law of total probability can be written as
\begin{equation}\label{eq::qwe1}
p(y_2|a_1,a_2) =
\int p(y_2|a_1,y_1,a_2) p(y_1|a_1)
\frac{p(a_2|a_1,y_1)}{p(a_2|a_1)} dy_1.
\end{equation}
In the causal interpretation, the model characterizes the density of the \emph{counterfactual} random variable $Y_2(a_1,a_2)$, which represents the value $Y_2$
would take if
$(A_1,A_2)$ had instead been equal to
$(a_1,a_2)$.
As we explain in Section \ref{sec::causal},
the density of the counterfactual $Y_2(a_1,a_2)$ 
is
\begin{equation}\label{eq::qwe2}
p(y_2(a_1,a_2)) = \int p(y_2|a_1,y_1,a_2) p(y_1|a_1) dy_1,
\end{equation}
and we see that
$p(y_2|a_1,a_2) \neq p(y_2(a_1,a_2))$.

{The problem with using augmented epidemic models both as DGMs to estimate their parameters, and as MSMs 
to produce downstream causal inferences such as scenario predictions or estimating causal effects, 
is that these inferences are inconsistent.}        
We provide three ways to fix this problem
(Methods 2, 3 and 4). 
{For completeness, we call Method 1 the approach we described above.}
\begin{itemize}
    \item[\bf Method 1.]
    The augmented epidemic model is used as a DGM for the observed outcome $Y$, and 
    its parameters are estimated by maximum likelihood or Bayes. 
    This may seem like the most natural approach
    but, as we will explain, it should be avoided because
    it leads to the $g$-null paradox and yields
    inconsistent causal estimators.
    The reason is that, in these models,
    correlation and causation are entangled
    and cannot be separated.
    \smallskip
    \item[\bf Method 2.]
    As in Method 1, the augmented epidemic model is used as a DGM.
    The causal effect is extracted from the DGM using the
    $g$-formula (see \cref{eq::g3}) and the model parameters are then estimated
    using estimating equations (see \cref{eq::ee}). 
    \smallskip
    \item[\bf Method 3.] 
    The augmented epidemic model is used as an MSM, 
    that is as a model for 
    the counterfactual $Y(a)$.
    The parameters are estimated using estimating equations.
    Method 3 is the simplest approach and is 
    standard in causal inference.
    \smallskip
    \item[\bf Method 4.]
    As with Method 3,
    we use the augmented epidemic model as an MSM for 
    the counterfactual $Y(a)$.
    Then, in contrast to Method 3, we construct a full DGM
    in a way that is consistent with the MSM.
    Then we can apply maximum likelihood
    to the full model and obtain consistent estimates. 
    In terms of model construction,
    this is the most technically challenging approach.
\end{itemize}
Methods 2, 3 and 4 all produce valid downstream causal inference. In our view, 
Method 3 is the simplest and most natural, and accords with common practice in causal inference.
{An other advantage of Method 3 is that it handles confounders more easily, 
because they are added to the propensity score, 
which is an ingredient in the estimating equations.
In contrast, for Methods 1, 2 and 4, confounders must be included in the DGM, 
which requires more modeling assumptions.}
Ultimately, we recommend Method 3
but we shall consider all four approaches.

\subsection{Related Work}

The literature on causal inference is vast.
Good references for background include
\citet{whatif,imbens2015,pearl2009}.
Of particular relevance is
\citet{robins2000marginal}
which defines MSMs.
We will mainly be concerned with causal inference
from a single time series because that's how most epidemic data arise; some of the challenges
in such settings have been discussed in 
\citet{cai2024causal}. 

The literature on epidemic modeling is also very large.
However, papers dealing with epidemic models
using explicit causal methods are less common.
\citet{halloran1995} deals explicitly with
infectious diseases in a causal framework
and considers violations of the ``no interference'' assumption
in which one subject's outcome can be affected by another
subject's intervention.
\cite{bhatt, flaxman2020} study the effect of various interventions
during the Covid pandemic using
the semi-mechanistic model in \cref{eq::bhatt,eq::bhatt_R}.
\citet{bonvini2022}
consider the effect of social mobility
on deaths for the Covid pandemic.
They use the approach described in this paper (Method 3).
\cite{ackley2017}
consider how causal graphs
might be used in the context of SIR models.
{\cite{feng2024}
study a particular causal method, known as differences in differences,
in the context of SIR models.
\cite{barrero2025} show how
causal structural models
can be used to help
formulate the study of the causal effects of climate
on infectious disease transmission.
\cite{callaway2023policy} consider multiple time series
corresponding to different locations
with a binary intervention that occurs at a single time.
A SIR model is used to motivate certain assumptions but the
method estimates the causal effect
by comparing cases in treated locations to the estimated outcome if they had not
been treated, using a regression model and propensity weights.
The estimator is quite simple
and does not rely heavily on the SIR model thanks to the information
from multiple locations and the simplicity of the intervention.}

Papers that discuss the $g$-null paradox problem
when treating time varying data generating models
as MSMs
include
\citet{robins1986new,rw,bates2022causal,robins2000b}.

{We emphasize that our focus in this paper is on adjustment methods,
that is, methods that adjust for observed confounders.
There are, of course, other approaches such as
differences-in-differences, regression discontinuity designs, 
instrumental variables, etc.}

\subsection{Paper Outline}

Section~\ref{sec::epi} presents examples of augmented epidemic models that have been used in the literature.
In \cref{sec::causal} we provide a brief
background for causal inference.
In \cref{sec::phantoms} we explain why the ubiquitous $g$-null paradox phenomenon 
arises when Method 1 is used and 
how Methods 2, 3 and 4, {described in~\cref{sec::extraction,sec::estimation,section::full}}, 
provide remedies to it. 
Finally, we present examples based on simulated and observational data 
in \cref{sec::example}, brush on the problem of model misspecification in 
\cref{sec::misspec} and conclude in \cref{sec::conclusion}.

{
The three main takeaways are:
(i) When assessing
the effect of interventions in times series,  
maximum likelihood and Bayes estimates of model parameters are typically inconsistent and subsequent causal effect estimates are inconsistent too.
Hence, MLEs and Bayes estimates should not be used.
(ii) Three alternative methods could be used instead. One of them, Method 3, is easier to implement because it requires fewer model assumptions.
(iii) In the context of Method 3, the epidemic model augmented with interventions provides an interpretable marginal structural model.
}

\subsection{{Definitions and Terminology}}

\subsubsection{Augmented Models}
Suppose we are given a baseline epidemic model
$p(\overline{y}_t;\zeta)$
which correctly describes the joint density of
$\overline{y}_t$ in the absence of an intervention.
An augmented model is a family of densities
$p(\overline{y}_t,\overline{a}_t;\zeta,\beta_A)$
with two properties:
\begin{itemize}
    \item[(1)] If $\beta_A=0$ then 
    $p(\overline{y}_t,\overline{a}_t;\zeta,\beta_A) = p(\overline{y}_t;\zeta)$.
    \smallskip
    \item[(2)] If $\overline{a}_t = (0,\ldots, 0)$ then
    $p(\overline{y}_t,\overline{a}_t;\zeta,\beta_A) = p(\overline{y}_t;\zeta)$,
    assuming that $a_t=0$ corresponds to no intervention.
\end{itemize}
These conditions imply that when there is no intervention we get back the
original model.
In particular, under the null hypothesis of no causal effect,
we have $\beta_A=0$ and the model
reduces to the baseline epidemic model.
We write $\theta = (\zeta,\beta_A)$ in what follows.

\subsubsection{Causal Graphs} 
We will sometimes use directed graphs to illustrate models
where arrows denote causal relationships.
In \cref{fig::dagsemi0}, $Y_t$ denotes an observed outcome (such as deaths), $A_t$ denotes the intervention of interest (such as public health policies),
and $X_t$ denotes confounders.
Latent variables are indicated with pink nodes.
Importantly, we allow phantom variables $U$.

\subsubsection{Phantom Variables}
Phantom variables are unobserved variables
that affect $Y_t$, and possibly all other variables, 
but that do not directly affect $A_t$, so they are not confounders. 
They were initially introduced by \citet{robins1986new} 
and later named ``phantoms'' by \citet{bates2022causal}. 
They play an important role in causal inference because
they are the source of 
the $g$-null paradox,
as we will explain in \cref{sec::phantoms}.

{Phantoms are ubiquitous.}
For example, air quality $U$ might affect deaths $Y$ from COVID-19 but will
not likely affect mobility $A$ (unless $U$ is extreme).
{Other factors like healthcare system capacity, pre-existing health conditions, 
population age distribution, viral strain characteristics, seasonality, 
socioeconomic status, and public health messaging may influence death rates 
without directly affecting behaviors such as mobility or mask-wearing.
In all these cases, the unobserved variables are not direct causes of 
treatment, but they shape outcome dynamics and can complicate causal 
interpretation in time-series settings.}

\begin{figure}[t!]
    \begin{center}
    

    \begin{tikzpicture}[->, shorten >=2pt,>=stealth, 
        node distance=1cm, noname/.style={ ellipse, minimum width=5em, minimum height=3em, draw } ]
        \node (0) {$\cdots$};
        \node (1) [right= of 0] {$X_{t-1}$};
        \node (2) [right= of 1] {$A_{t-1}$};
        \node (3) [right= of 2] {$Y_{t-1}$};
        \node (4) [right= of 3] {$X_t$};
        \node (5) [right= of 4] {$A_t$};
        \node (6) [right= of 5] {$Y_t$};
        \node (7) [right= of 6] {$\cdots$};
    
        \node[circle,fill=red!20,draw,minimum size=30pt,inner sep=0pt] (8) [below= of 2] {$U_{t-1}$};
        \node[circle,fill=red!20,draw,minimum size=30pt,inner sep=0pt] (9) [below= of 5] {$U_t$};
    
        \path (0) edge [thick] node {} (1);
        \path (1) edge [thick] node {} (2);
        \path (1) edge [thick, bend left=50pt] node {} (3);
        \path (1) edge [thick, bend left=50pt] node {} (4);
        \path (1) edge [thick, bend left=50pt] node {} (5);
        \path (1) edge [thick, bend left=50pt] node {} (6);
        \path (2) edge [thick] node {} (3);
        \path (2) edge [thick, bend left=50pt] node {} (4);
        \path (2) edge [thick, bend left=50pt] node {} (5);
        \path (2) edge [thick, bend left=50pt] node {} (6);
        \path (3) edge [thick] node {} (4);
        \path (3) edge [thick, bend left=50pt] node {} (5);
        \path (3) edge [thick, bend left=50pt] node {} (6);
        \path (4) edge [thick] node {} (5);
        \path (4) edge [thick, bend left=50pt] node {} (6);
        \path (5) edge [thick] node {} (6);
        \path (6) edge [thick] node {} (7);
        
        \path (8) edge [thick] node {} (1);
        \path (8) edge [thick] node {} (3);
        \path (8) edge [thick] node {} (9);
        \path (9) edge [thick] node {} (4);
        \path (9) edge [thick] node {} (6);
    \end{tikzpicture}
            
    %

    
    
    \end{center}

\caption{\sl {\bf Example DAG for epidemic data. } 
The arrows indicate possible
causal relationships between the outcome $Y$, 
intervention $A$ and confounders $X$.
Latent variables $U$ are in pink;
$U$ does not directly affect $A$ -- 
we say that $U$ is a phantom variable.
If there were arrows from $U$ to $A$ then $U$ would instead be a confounder.}
\label{fig::dagsemi0}
\end{figure}

\section{{Background on Epidemic Models} \label{sec::epi}}

{We review some epidemic models
with attention to how interventions have been, or can be, included. 
In~\cref{sec::estimation} we argue that such augmented 
models can be used as MSMs, which is a type of causal models for
time varying inference that describe how the intervention 
affects the outcome.}

\begin{example}[SIR Model]
{
Perhaps the most fundamental epidemic model
is the SIR model
(Susceptible-Infected-Recovered)
due to \citet{kermack1927contribution}.
The model is given by three differential equations
\begin{equation}\label{eq::SIR}
\begin{aligned}
\frac{dS_t}{dt} &= -\frac{\alpha I_t S_t}{N}, \\
\frac{dI_t}{dt} &= \frac{\alpha I_t S_t}{N}-\gamma I_t, \\
\frac{dR_t}{dt} &= \gamma I_t,
\end{aligned}
\end{equation}
for $t >0$, where $S_t$, $I_t$ and $R_t$ are the numbers of susceptibles,
infected, and removed (deaths or recovered) 
at $t$, $N =S_t + I_t +R_t$ is the total population size, 
and $\alpha$ and $\gamma$ are the rates of infection and of removal, respectively.
To study the effect of an intervention on the outcome,
$\alpha$ could be replaced by
$\alpha_t = \alpha e^{\beta_A A_t}$,
where $\beta_A$ is the parameter that modulates 
the effect of $A_t$ on the subsequent number of infections $I_t$.
Confounders $X_t$ can be added in a similar way, for example, by setting
$\alpha_t = \alpha e^{\beta_A A_t + \beta_X^T X_t}$.
}

{
\cite{n2022effect} include a covariate $X_t$ in a SIR model 
by modeling the rate of infection as 
$$
\alpha_t = \frac{\beta_0 }{1+\gamma I X_t},
$$
where $X_t$
measures the level of adherence to disease control measures at time $t$. 
}

{
\cite{chernozhukov2020},
evaluated the effects of
policies on Covid 19 using a system
of linear structural equation models
which included
policies and  behaviors.
They did not directly use an epidemic model.
Rather, they use properties of the SIR model as motivation
the form of their model.
}
\end{example}

{
\begin{example}[Discrete SEIR Model]
\label{ex::seir}
\cite{lekone2006statistical,gibson1998estimating,mode2000stochastic}
proposed a discrete SEIR model
which models
the numbers of susceptibles $S_t$,
exposed $E_t$,
infected $I_t$
and the cumulative number of removed up to time $t$, $R_t$, by
\begin{equation} \label{eq::seir}
\begin{aligned}
    S_{t+h} &= S_t - B_t, \\
    E_{t+h} &= E_t + B_t - C_t, \\
    I_{t+h} &= I_t + C_t - D_t, \\
    R_{t+h} &= R_t + D_t, 
\end{aligned}
\end{equation}
where
$h$ represents the time interval
(e.g. $h=1$ day),
$B_t \sim \text{Binomial}(S_t,p_{B,t})$ is the number of susceptibles
who become infected,
$C_t \sim \text{Binomial}(E_t,p_C)$ is the number of new cases and  
$D_t \sim \text{Binomial}(I_t,p_D)$ is the number of newly removed cases.
The parameters are expressed as 
$$
p_{B,t} = 1 - \exp\Bigl\{ - \frac{\eta_t}{N} h I_t\Bigr\},\ \ 
p_C = 1 - e^{-\rho h}, \ \ p_D = 1 - e^{-\gamma h},
$$
where $\eta_t$ is the time-dependent transmission rate, $1/\rho$ is the mean incubation period, $1/\gamma$ is the mean infectious period, 
and $S_t + E_t + I_t + R_t = N$ is the total population size.
Again, we might observe $C_t$ or $I_t$ -- the other variables being latent -- 
or we might observe a variable $Y_t$ related to $I_t$.
These papers did not explicitly include interventions
but to include an intervention $A_t$,
we could
replace $\eta_t$ with 
$\eta(\bar{A}_t; \beta) = e^{- \beta_0 - \beta_A A_t}$.
\end{example}
}

{
\begin{example}
\cite{grigorieva2015optimal}
use a modified SEIR model
\begin{align*}
\dot{S}(t) &= -N^{-1} [u(t)I(t)+ v(t)E(t)] S(t)\\
\dot{E}(t) &= N^{-1} [u(t)I(t)+ v(t)E(t)] S(t)-\delta E(t)\\
\dot{I}(t) &= \delta E(t)-\gamma I(t)
\end{align*}
where
$u(t)$ and $v(t)$ model continuous time effects of interventions 
on infected and on exposed.
\end{example}
}

{
\begin{example}
\cite{GIFFIN2022100711}
consider a discretized spatial, SIR model
\begin{align*}
S_j(t+1) - S_j(t) &= -\lambda_j(t)\\
I_j(t+1) - I_j(t) &= \lambda_j(t)-\gamma I_j(t)\\
R_j(t+1) - R_j(t) &= \gamma I_j(t)
\end{align*}
where
$$
I_j(t) = \beta_j(t)\frac{S_j(t)}{N}\sum_{k=1}^K W_{jk}I_k(t),
$$
$W_{jk}$ is the contact rate between regions $j$ and $k$
and
$$
\log \beta_j(t) =\alpha_0 +X_j(t)^T \alpha_1 + A_j(t) \delta
$$
where $A_j(t)$ is mobility and are $X(t)$ covariates.
\end{example}
}

\section{Background on Causal Inference\label{sec::causal}}

Putting aside epidemic models
for a moment, we now review some background
on causal inference.

First, consider a single outcome $Y$ and
a binary intervention $A \in \{0,1\}$. 
The {\em counterfactual} $Y(a)$ is the value
the outcome $Y$ would take if
the intervention $A$ were set to $a$.
Thus,
we now have four random variables
$(A,Y,Y(0),Y(1))$
where
$Y(0)$ is the value $Y$ would have if $A=0$
and
$Y(1)$ is the value $Y$ would have if $A=1$.
The counterfactuals $Y(0)$ and $Y(1)$ are linked
to the observed data $(A,Y)$ by the equation
$Y= Y(A)$.
If $A=1$ then $Y = Y(1)$ but $Y(0)$ is unobserved.
If $A=0$ then $Y = Y(0)$ but $Y(1)$ is unobserved.
Many causal questions are quantified by
these counterfactuals.
For example,
$\E[Y(1)] - \E[Y(0)]$
is used to quantify the
causal effect
of the intervention.
It can be shown that
if there are no confounding variables ---
variables that affect both $Y$ and $A$ ---
then
$P(Y\leq y|A=a) = P(Y(a)\leq y)$
so that the distribution of
the counterfactual $Y(a)$
is the same as 
the conditional distribution
of the observable $Y$ given $A$,
and thus $\E[Y(a)] = \E[Y | A=a]$.
But if there are confounding variables $X$ then
$P(Y\leq y|A=a) \neq P(Y(a)\leq y)$ and 
thus $\E[Y(a)] \neq \E[Y | A=a]$.
In this case, 
it can be shown 
(under Conditions C1, C2 and C3, described below)
that
\begin{equation}
P(Y(a) \leq y) = \int P(Y\leq y|A=a,X=x) dP(x).
\label{eq::babyg}
\end{equation}
Thus we can derive the distribution of $Y(a)$
from the distribution for $(X,A,Y)$
using \cref{eq::babyg}.

Now consider observed time series data of the form
$$
(X_1,A_1,Y_1), \ldots, (X_T,A_T,Y_T),
$$
where
$A_t$ is some intervention at time $t$,
$Y_t$ is the outcome of interest at $t$
and $X_t$ refers to potential confounding variables
which might affect $A_t$ and $Y_t$
(and future values).
As above, 
the causal effect of the intervention can be quantified by the counterfactual 
$Y_t(\overline{a}_t)$,
which is the value $Y_t$
would have if
a hypothetical intervention sequence was 
$\overline{a}_t = (a_1,\ldots, a_t)$
rather than the actual observed sequence
$\overline{A}_t = (A_1,\ldots, A_t)$.
For example, suppose that
$a_t=1$ means that there is a mandate to wear masks
and
$a_t=0$ means that there is no mask mandate.
Then
$Y_t(0,0,\ldots, 0)$
is the outcome at time $t$ if there was never a mask mandate.
In some literature,
$\E[ Y_T(\overline{a}_T)]$ is denoted by
$\E[Y_T| {\rm do}(\overline{a}_T)]$.

Confounders are generally present in the time series setting
because $Y_s$ typically affects future values $Y_t$ 
and $A_t$ for $t>s$, as depicted in \cref{fig::dagsemi0}.
Then
causal and association effects are different,
e.g.
$
\E[Y_t(\overline{a}_t)] \neq \E[Y_t| \overline{A}_t = \overline{a}_t].
$
\citet{robins1986new} proved that
$$
\E[Y_t(\overline{a}_t)] = \psi(\overline{a}_t),
$$
where
\begin{equation}\label{eq::g3}
    \psi(\overline{a}_t) \equiv \int \cdots \int \E[ Y_t | \overline{x}_{t},\overline{y}_{t-1},\overline{a}_{t}] 
    \prod_{s=1}^t p( x_s,y_s | \overline{x}_{s-1},\overline{a}_{s-1},\overline{y}_{s-1}) dx_s dy_s,
\end{equation}
provided 
Conditions (C1)-(C3) below are met.
In what follows,
we will often write
$\psi(\overline{a}_t; \theta)$
where $\theta$ denotes any parameters that are involved.
\cref{eq::g3} is known as the
{\em $g$-formula} for the mean causal effect,
but there are similar expressions
for densities, cdf's, quantiles etc.
In particular, 
let $p_{\overline{a}_t}(y_t)$ 
denote the density of counterfactual $Y_t(\overline{a}_t)$
evaluated at $y_t$. Then
\begin{equation}\label{eq::g4}
    p_{\overline{a}_t}(y_t)=
    \int \cdots \int p(y_t | \overline{x}_{t}, \overline{a}_t, \overline{y}_{t-1})
    \prod_{s=1}^t p(x_s, y_s | \overline{x}_{s-1}, \overline{a}_{s-1}, \overline{y}_{s-1}) dx_s dy_s,
\end{equation}
which extends \cref{eq::babyg} to the time series setting.

The $g$-formula has a graphical interpretation.
Starting with a directed graph $G$ such as \cref{fig::dagsemi0},
form a new graph $G^*$ in which all arrows pointing into any $A_s, s \leq t,$ are removed
and in which any $A_s$ is fixed at a value $a_s$; see \cref{fig::intervene}.
\cref{eq::g4} is then the marginal density for $Y_t$
corresponding to the density in the graph $G^*$.

\begin{figure}[t!]
    \begin{center}
    \begin{tikzpicture}[->, shorten >=2pt,>=stealth, 
        node distance=1cm, noname/.style={ ellipse, minimum width=5em, minimum height=3em, draw } ]
        \node (0) {$\cdots$};
        \node (1) [right= of 0] {$X_{t-1}$};
        \node (2) [right= of 1] {\color{blue} $a_{t-1}$};
        \node (3) [right= of 2] {$Y_{t-1}$};
        \node (4) [right= of 3] {$X_t$};
        \node (5) [right= of 4] {\color{blue} $a_t$};
        \node (6) [right= of 5] {$Y_t$};
        \node (7) [right= of 6] {$\cdots$};
    
        \node[circle,fill=red!20,draw,minimum size=30pt,inner sep=0pt] (8) [below= of 2] {$U_{t-1}$};
        \node[circle,fill=red!20,draw,minimum size=30pt,inner sep=0pt] (9) [below= of 5] {$U_t$};
    
        \path (0) edge [thick] node {} (1);
        \path (1) edge [thick, bend left=50pt] node {} (3);
        \path (1) edge [thick, bend left=50pt] node {} (4);
        \path (1) edge [thick, bend left=50pt] node {} (6);
        \path (2) edge [thick] node {} (3);
        \path (2) edge [thick, bend left=50pt] node {} (4);
        \path (2) edge [thick, bend left=50pt] node {} (6);
        \path (3) edge [thick] node {} (4);
        \path (3) edge [thick, bend left=50pt] node {} (6);
        \path (4) edge [thick, bend left=50pt] node {} (6);
        \path (5) edge [thick] node {} (6);
        \path (6) edge [thick] node {} (7);
        
        \path (8) edge [thick] node {} (1);
        \path (8) edge [thick] node {} (3);
        \path (8) edge [thick] node {} (9);
        \path (9) edge [thick] node {} (4);
        \path (9) edge [thick] node {} (6);
    \end{tikzpicture}
            
    \end{center}
    \caption{\sl {\bf Intervention graph} from \cref{fig::dagsemi0} after setting
    $\overline{A}_t=\overline{a}_t$.}
    \label{fig::intervene}
    \end{figure}

The $g$-formula is valid under three conditions:
\begin{itemize}
    \item[(C1)] 
    No interference:
    if
    $\overline{A}_t =\overline{a}_{t}$ 
    then
    $Y_t(\overline{a}_t) = Y_t$.    
    \item[(C2)] 
    Positivity:
    there exists $\epsilon > 0$ such that $\pi(a_t| \overline{x}_{t}, \overline{a}_{t-1}, \overline{y}_{t-1}) > \epsilon$
    for all values of $\overline{x}_t$, $\overline{a}_t$ and $\overline{y}_{t-1}$,
    where
    $\pi(a_t| \overline{x}_{t},\overline{a}_{t-1}, \overline{y}_{t-1})$
    is the density of $A_t$ given the past.    
    \item[(C3)]
    No unmeasured confounding:
    the variable
    $Y_t(\overline{a}_t)$ is independent of
    $A_t$ given the past measured variables.
\end{itemize}
Condition (C1) means that the observed $Y_t$
is equal to the counterfactual
$Y_t(\overline{a}_t)$
if the observed intervention sequence
$\overline{A}_t$
happens to equal $\overline{a}_t$.
This means a subject's outcome is affected
by their intervention but not affected
by another subject's intervention.
{This assumption can be violated if there are spillover
effects from one subject to another.
For example, in the case of vaccines, a subjects outcome (infection)
can be affected by the vaccine status of other individuals.
The no interference assumption is reasonable
when such spillover are considered unlikely.
In cases where the outcome is a geographic region
(such as a county), no interference is reasonable 
if the main effect is expected to be local.}
Condition (C2) means that, conditional on the past, every subject has nonzero probability
of receiving intervention at any level.
Condition (C3) means that we have measured
all important confounding variables, which are variables that
affect the intervention and the outcome.

Alternatively, one can specify a simple model for
counterfactual $Y_t(\overline{a}_t)$ directly instead of applying the $g$-formula,
akin to specifying a regression model for the effect of $\overline{a}_t$ on $Y_t$.
\cref{eq::MSM0} provides an example.
This is called a {\em marginal structural model}
(MSM; \citealp{robins2000marginal}).
This approach is semi-parametric, in the sense that we do not need
a model for the joint distribution of the time series, 
$p(y_t | \overline{x}_{t}, \overline{a}_t, \overline{y}_{t-1})$,
since we do not need to apply the $g$-formula 
in \cref{eq::g3} or \cref{eq::g4}.
However, the trade-off for simplicity is that MSMs often fail 
to incorporate domain-specific knowledge.
In contrast, the approach we advocate for in~\cref{sec::estimation} (Method 3) 
includes knowledge about 
underlying epidemic dynamics
by interpreting augmented epidemic models as 
causal densities $p_{\overline{a}_t}(y_t;\theta)$.
That is, 
we treat augmented epidemic models as MSMs for the counterfactual $Y_t(\overline{a}_t)$.

Next we turn to the problem of parameter estimation.

\section{Method 1: Maximum Likelihood Estimation and the Problem of Phantom Bias}
\label{sec::phantoms}

Method 1 consists of using maximum likelihood or Bayesian methods
to estimate the model parameters
(estimates are 
based on the likelihood function so the model is implicitly thought of as a DGM for $Y$) 
then using the fitted model to compute the causal effect of $A$ 
on $Y$ or make scenario predictions for $Y(a)$
(so the model is now implicitly thought of as a causal model for $Y(a)$).
Then we run into a problem
first identified by \cite{robins1986new}:
maximum likelihood and Bayes estimates
of causal effects are inconsistent.
\cite{robins1986new} called this the $g$-null paradox
because the effect is especially pernicious in the null case, 
when there
is no causal effect but the estimated causal effect will be nonzero.

\begin{figure}[b!]
\begin{center}
\begin{tikzpicture}[->, shorten >=2pt,>=stealth, 
node distance=1cm, noname/.style={ ellipse, minimum width=5em, minimum height=3em, draw } ]
\node[] (1) {$A_0$};
\node (2) [right= of 1] {$X_1$};
\node (3) [right= of 2] {$A_1$};
\node (4) [right= of 3] {$Y_1$};
\node[circle,fill=red!20,draw] (5) [below= of 3] {$U$};
\path (1) edge [thick] node {} (2);
\path (1) edge [thick, bend left=50pt] node {} (3);
\path (2) edge [thick] node {} (3);
\path (5) edge [thick] node {} (2);
\path (5) edge [thick]  node {} (4);
\end{tikzpicture}
\end{center}
\caption{ \sl {
{\bf Effect of phantoms.}
The latent phantom variable $U$ is not a confounder because it has no arrows to $A_0$ or $A_1$.
Neither $A_0$ nor $A_1$ have a causal effect on $Y_1$.
The variable $X_1$ is a collider, meaning that two arrowheads point to $X_1$.
This implies that $Y_1$ and $(A_0, A_1)$  are dependent conditional on $X_1$,
which in turn implies that 
the parameters that relate $Y_1$ to $(A_0, A_1)$ in the epidemic model
will be non-zero even though there is no causal effect.}}
\label{fig::nulldag}
\end{figure}

\begin{example}[Synthetic example] 

{
To illustrate this, assume that we have data
that conforms with the directed graph in \cref{fig::nulldag}, 
from \citet{rw} Section 1.3. To estimate the causal effect 
of $A_0$ or $A_1$ on $Y_1$,
suppose that we assume the model
\begin{align*}
A_0   & \sim p(a_0),\\
X_1 & \sim \mathrm{Bernoulli}(\mathrm{expit}(\xi_0 + \xi_1 A_0)), \\
A_1   & \sim p(a_1|I_1,A_0),\\
Y_1 &= \beta_0 + \beta_1 A_0 + \beta_2 X_1 + \beta_3 A_1 + \delta,
\end{align*}
where $\delta$ are mean 0 Normal random variables
and $p$ is a logistic model,
to which we apply the $g$-formula to obtain 
the causal effect on $Y_1$ of setting
$(A_0,A_1)$ to $(a_0,a_1)$:
$$
\psi(a_0,a_1)=\E[Y_1(a_0,a_1)] = \beta_0 + \beta_1 a_0 + \beta_2 \mathrm{expit}(\xi_0 + \xi_1 a_0) + \beta_3 a_1.
$$
Note that $U$ did not appear in the model because it is not observed.}

{
The true value of the causal parameter
$\beta_A = (\beta_1,\beta_3)$ should be $(0,0)$ 
since there is no path from $A_0$ or $A_1$ to $Y_1$ in \cref{fig::nulldag}, 
and thus no causal effect of $(A_0, A_1)$ on $Y_1$.
However, \citet{robins1986new,rw} showed that 
the maximum likelihood estimators $\hat\beta_1$ and $\hat\beta_3$ are not zero
and in fact converge to nonzero numbers in the large sample 
limit. Consequently, the estimated causal effect depends on $(A_0, A_1)$, 
even though $(A_0, A_1)$ has no effect on $Y_1$.
This happens because there is a phantom $U$ 
that affects $X_1$ and $Y_1$, 
which makes $X_1$ a collider on the path $Y_1, U, X_1, A_0, A_1$ 
and in turn induces conditional dependence between 
$Y_1$ and $(A_0,A_1)$.
}
\end{example}

\begin{example}[Semi-mechanistic Hawkes model]
\label{ex::semi2}
Consider the semi-mechanistic model in \cref{eq::bhatt}, with reproduction number 
$R(\bar{A}_t, \beta)$ in \cref{eq::bhatt_R}. If confounders $X_t$ exist, we could 
include them in 
$R(\bar{A}_t, \beta)$ by adding the additive term $\beta_X X_t$
    \begin{equation} \label{eq::Rconfound}
    R(\bar{A}_t, \beta) = \frac{K}{1 + \exp(\beta_0 + \beta_X X_t +\beta_A A_t)},
    \end{equation}    
which is a standard strategy in the regression set-up. 
The ML estimates of the parameters are not available in closed form but can be obtained numerically \citep{bong2024}. 
The simulation study in \cref{sec::simulation_semi} \cref{fig::conf_int_ML_bhatt} shows
that the ML estimate of the causal effect 
$\beta_A$ in \cref{eq::Rconfound} is biased. In particular, when there is no causal effect ($\beta_A=0$),
the ML estimate is significantly different from zero.
\end{example}

{
In the two examples above, 
the phantoms create a situation where there is no causal effect, yet the 
estimated causal effect is nonzero. 
This happens because sequentially specified non-linear parametric 
models are not variation independent: 
causation and conditional dependence between outcome and treatment induced 
by phantoms are tied together
in the parameterization of the model. 
Essentially, there are not enough parameters to model the two effects separately.
Moreover, the ML and Bayes parameter estimates
are driven strongly by the conditional dependence
rather than by the causal effect.
Then both the causal and dependence 
estimates will be nonzero even though 
there is no causal effect.
}
{To summarize, letting $\gamma$ represent some measure of conditional 
dependence and $\beta_A$ denote the causal parameter, we have 
$$
\underbrace{\beta_A = 0 \ {\rm but} \ \gamma \neq 0 }_{\rm due \ to \ phantoms} \ \Longrightarrow\ 
\underbrace{\hat\gamma\neq 0}_{\rm due \ to \ estimation \ based \ on \ likelihood}  \ \Longrightarrow\ 
\underbrace{\hat\beta_A \neq 0}_{\rm due \ to \ variation\ dependence}.
$$
}
{That is, the $g$-null paradox is due to three things:
phantoms, which enable conditional dependence between outcome and treatment 
even though there is no causal effect,
variation dependence between causal and dependence effects,
and the fact that the 
causal estimate is nonzero due to dependence.
}

The fact that in the real world we can have
dependence but no causal effect and the model
cannot represent this, means that the model
is misspecified.
If $\Theta_0$ denotes the parameter values that correspond to no causal effect
and
$\Theta_+$ denotes the parameter values that correspond to conditional dependence,
we have that
$\Theta_0 \bigcap \Theta_+ = \emptyset$.
This was first pointed out by
\citet{robins1986new}
and has received much attention
since then; see, for example,
\citet{rw}, \citet{bates2022causal},
\citet{robins2000b}, \citet{babino2019multiple} and \citet{evans2024}.
It appears that the problem has gone unnoticed in the
literature on modeling epidemics.

The one case where phantoms do not induce a $g$-null paradox is when all of the
equations in the model are linear.
This is rarely the case in epidemic models.
Generally, any finite dimensional parametric model
which models each variable given the past and
has some non-linear component and will suffer the $g$-null paradox.
An infinite dimensional (nonparametric) model would solve the problem,
but this is not practical when we have many time points due to the curse  
of dimensionality. This is especially true with a single time series model.
There are more complicated parametric models that avoid
the problem as we describe in Method 4.

Unlike Method 4, Methods 2 and 3 do not require a more complicated parametric model.
They avoid the $g$-null paradox
by estimating the model parameters using 
estimating equations rather than the likelihood function.

{We conclude this section by emphasizing that phantom 
bias is different than unmeasured confounder bias.
Confounders are variables that affect $Y_t$ and $A_t$.
Unobserved confounders render the causal effect unidentifiable.
Phantoms do not affect $A_t$ so they do not change the $g$-formula, 
and they do not render the causal effect unidentifiable.
Rather, they cause the maximum
likelihood estimate or Bayes estimate for any sequentially specified non-linear
parametric model
to be inconsistent.}

\section{Method 2: Causal Effect Extraction and Estimating Equations} 
\label{sec::extraction}

If the epidemic model is treated as a DGM, then
we can avoid the $g$-null paradox by using the following workflow:

(1) If there are confounders $X_t$, include them in the DGM.

(2) Extract the causal effect $\psi(\overline{a}_t;\theta)$ from the DGM
using the $g$-formula in \cref{eq::g3}.
(Usually, the $g$-formula is intractable and needs to be computed
by simulation.) 

(3) Estimate the parameters of the MSM using estimating equations.
Specifically, 
\cite{robins2000marginal}
showed that
$\theta$ satisfies
\begin{equation} \label{eq::ee0}
\sum_t\E\left[ 
h_t(\overline{A}_t) (Y_t - \psi(\overline{A}_t;\theta)) W_t\right] = 0
\end{equation}
where
\begin{equation} \label{eq::pp0}
    W_t=
\prod_{s=1}^t
\frac{\pi(A_s| \overline{A}_{s-1})}
{\pi(A_s| \overline{A}_{s-1},\overline{X}_{s},\overline{Y}_{s-1})}.
\end{equation}
Here,
$\pi(A_s| \overline{A}_{s-1},\overline{X}_{s},\overline{Y}_{s-1})$ ---
called the {\em propensity score} ---
is the density of $A_s$ given the whole past and
$\pi(A_s| \overline{A}_{s-1})$
is the density of $A_s$ given the past $A$'s. Both these densities
are derived from the DGM.  
The term $h_t(\overline{A}_t)$
is an arbitrary vector of dimension dim($\theta$) 
of
functions of $\overline{A}_t$.
The simplest choice is $h_t(A_1,\ldots, A_t) = 1$ when dim$(\theta)=1$.
The choice of $h_t(A_1,\ldots, A_t)$ affects the variance of the estimator
but any choice leads to
consistent estimates of $\theta$.
In principle, there 
is an optimal choice that leads to the smallest possible variance but 
constructing it can be difficult \citep{robins2000b, kennedy2015}.
For simplicity in the three examples of \cref{sec::example}, we used 
$h_t(\overline{A}_t) = (1, A_t)$ to estimate the two-dimensional parameters
of the MSM.

We define $\hat\theta$ to be the solution to the sample version of \cref{eq::ee0}, namely,
\begin{equation}\label{eq::ee}
\sum_t 
h_t(\overline{A}_t) (Y_t - \psi(\overline{A}_t;\hat\theta) ) \hat W_t = 0
\end{equation}
where
$$
\hat W_t=
\prod_{s=1}^t
\frac{\hat \pi(A_s| \overline{A}_{s-1})}
{\hat \pi(A_s| \overline{A}_{s-1},\overline{X}_{s},\overline{Y}_{s-1})},
$$
with
``hats" signifying estimates.

\medskip

{\bf Confidence Intervals. \, } 
Under some regularity conditions,
we have 
$$
\sqrt{T}(\hat\theta - \theta)\rightsquigarrow N(0,\Sigma)
$$
with $\Sigma$ estimated by \citep{newey1987simple}
\begin{equation*}
    \frac{1}{T} \left[\tsum_t \nabla_\theta \hat\phi_t(\hat\theta) \right]^\top
    \left[\sum_{t, s} w(t, s) \cdot \hat\phi_t(\hat\theta) \hat\phi_t(\hat\theta)^\top \right]^{-1}
    \left[\tsum_t \nabla_\theta \hat\phi_t(\hat\theta) \right],
\end{equation*}
where $w$ is a smoothing kernel function and 
$$
\hat\phi_t(\hat\theta) \equiv 
h_t(\overline{A}_t) (Y_t - \psi(\overline{A}_t;\hat\theta)) W_t,
$$
from which we can construct confidence intervals for $\theta$ and 
the causal effect.
See also
\cite{dejong1997, dejong2000, andrews1991, andrews1988} for other estimates of $\Sigma$.
The conditions needed for the central limit theorem
involve mixing conditions which require that 
correlations in the data
eventually die off over time. 
{There is burgeoning
research on simulation based inference
\citep{tomaselli2025} for situations where asymptotic 
approximations for confidence intervals are not reliable, although this has yet to be
applied to causal inference.
}

\smallskip

As is evident in \cref{eq::g3}, calculating the causal effect $\psi(\overline{A}_t;\theta)$ requires 
the full joint distribution of $(\overline{X}_{t},\overline{A}_{t},\overline{Y}_{t})$.
While possible, postulating a reasonable joint distribution is a 
daunting task that entails making many assumptions. {Additionally, 
the
counterfactual model derived from the DGM is not easily interpretable.}
Method 3 is related 
to Method 2, but relies on an easily interpretable counterfactual model
and requires fewer
distributional assumptions, making Method 3 more practically attractive.

\section{Method 3 - Marginal Structural Models and Estimating Equations} \label{sec::estimation}

In this approach, {we do not treat the augmented epidemic model as a DGM and
we do not apply steps (1) and (2) of the workflow above. Instead,}
we interpret the augmented epidemic model
as an MSM, that is a model for the counterfactual $Y_t(\overline{a}_t)$.
This means that 
we fix $\overline{A}_t = \overline{a}_t$ 
and then the model gives the distribution of
$Y_t(\overline{a}_t)$.
In general, it is unlikely that we will be able to derive
a closed form for the distribution of
$Y_t(\overline{a}_t)$ because epidemic models are often hierarchical -- see, for example, 
\cref{eq::bhatt,eq::bhatt_R} and \cref{sec::epi} -- but we can use Monte 
Carlo simulation instead. 
Fix
$\overline{A}_t$ at $\overline{a}_t$, simulate $Y_1,\ldots, Y_t$
from the epidemic model parameterized by $\theta$, 
repeat this simulation $N$ times
giving values
$Y_{1}^{(k)},\ldots, Y_{t}^{(k)}$
for $k=1,\ldots, N$. Then the histogram of $Y_{t}^{(k)}$, 
$k=1,\ldots, N$, approximates the 
distribution of
$Y_t(\overline{a}_t)$, and the causal effect -- if taken 
to be the mean of that distribution -- is approximated with 
\begin{equation}\label{eq::monte_carlo}
\psi(\overline{a}_t;\theta) \approx
\frac{1}{N}\sum_{k=1}^N Y_{t}^{(k)}.
\end{equation}

{Having obtained the causal effect $\psi(\overline{a}_t;\theta)$, we
now apply the same step (3) of the workflow for Method 2. That is, 
we estimate $\theta$ by solving 
the estimating equations in 
\cref{eq::ee} and obtain a confidence interval for $\theta$ 
as specified there.
We already discussed the choice of $h_t$ in \cref{eq::ee}.
As for the propensity score $\pi(A_s| \overline{A}_{s-1},\overline{X}_{s},\overline{Y}_{s-1})$,
it is derived from the full joint 
distribution of $(\overline{X}_{t},\overline{A}_{t},\overline{Y}_{t})$ when we use Method 2,
but we must specify a model for it here since we do not assume a model for
the full joint 
distribution.}

\subsection{Propensity Score Model} \label{sec::propen}
{
When $A_t$ is continuous, a natural choice
is an autoregressive model of order $k$
$$
A_t =  \chi_0 + \chi_1 A_{t-1} + \cdots + \chi_{k} A_{t-k} +
\gamma_0^T X_{t} + \cdots + \gamma_k^T X_{t-k} + \lambda_1 Y_{t-1} + \cdots + \lambda_{k} Y_{t-k} + \epsilon_t
$$
where $\epsilon_1,\ldots, \epsilon_T \sim N(0,\sigma^2)$.
When $A_t$ is binary, 
a tractable model for the propensity score is the logistic autoregression
given by
$$
\pi(A_t| \overline{A}_{t-1},\overline{X}_t) = \pi_t^{A_t} (1-\pi_t)^{1-A_t}
$$
where
$$
\mathrm{logit}(\pi_t) = 
\chi_0 + \chi_1 A_{t-1} + \cdots + \chi_{k} A_{t-k} +
\gamma_0^T X_{t} + \cdots + \gamma_k^T X_{t-k} + \lambda_1 Y_{t-1} + \cdots + \lambda_{k} Y_{t-k}. 
$$
The parameters can be estimated by maximum likelihood, and
the maximum lag $k$ can be estimated using criteria like AIC or BIC.}

{A challenge is that the estimated propensity score
$\hat\pi(A_s| \overline{A}_{s-1},\overline{X}_{s},\overline{Y}_{s-1})$
can sometimes get very small,
which causes $\hat W_t$ and the resulting causal estimates
to have large variances.
Small propensity scores are often a problem in causal
inference and these problems are exacerbated in
time series models since we have to
multiple many weights together in \cref{eq::pp0}.
Dealing with small propensity scores
in causal inference is an active area of research.
One possibility is to use a method
called residual balancing \citep{Wodtke2020}
rather than estimating the propensity score.
This method uses moment conditions
to estimate the weights $W_t$ directly
and leads to more stable estimates.
The method was used successfully in \cite{bonvini2021causal}.
For a review on propensity scores see \cite{pan2018}.
}

\subsection{Challenges for Short Time Series}
{
Estimating the propensity score from a single time series
is challenging. 
All of the above assumes that we have enough observations $T$ to estimate
the parameters well.
The autoregressive model for the propensity score might have to allow
the maximum lag $k$ to be large if interventions in the far past effect current treatment.
Estimating the parameters in that case requires a long observation window.
Furthermore, the confidence intervals depend
on asymptotic Normality which also requires $T$ to be large.
If $T$ is not large then we may not be able to estimate the
parameters and asymptotic Normality may fail.
In these cases, causal inference might be infeasible.
If there are multiple time series then the situation might be better.
For example, suppose we have a time series for every county in the USA.
Then, estimation and inference become feasible again since we have much more information.
However, in this case, one might want to allow the parameters to vary across
counties which suggests using a hierarchical model.
Causal inference with hierarchical models
has been studied; see, for example,
\cite{feller2015} and \cite{weinstein2024}.
One possibility to estimate the parameters separately for each county,
shrink the estimates towards each other then
form debiased confidence intervals as in
\cite{armstrong2022robust,bong2024}.}

\subsection{Examples} \label{eq::examples}

We finish \cref{sec::estimation} with a rare example for which
the causal effect $\psi(\bar{a}_t, \theta)$, if taken to be the mean of the causal distribution, 
can be
computed in closed form instead of approximated by Monte Carlo, as in (\ref{eq::monte_carlo}).
For this example we also provide the derivatives of the estimating equations in \cref{eq::ee}, which are
needed to obtain $\hat \theta$ numerically.
We will use this example in \cref{sec::simulation_semi} to illustrate phantom bias.

We also consider the SEIR model, which does not admit a closed form solution for 
$\psi(\bar{a}_t, \theta)$ and must approximated by Monte Carlo in (\ref{eq::monte_carlo}).
We provide Monte Carlo derivatives of this approximation, so that we can use 
this model to further illustrate phantom bias 
in \cref{sec::simulation_seir}.

\begin{example}[Semi-mechanistic Model]
\label{ex::bhatt}
    Consider the model in \cref{eq::bhatt},
    with reproduction number 
    in \cref{eq::bhatt_R}.
    A different version takes
    \begin{equation} \label{eq::exp_semi}
\begin{aligned}
\Exp[I_t|\bar{A}_t, \bar{I}_{t-1}, \bar{Y}_{t-1}] &= \sum_{s<t} e^{\beta_0 + \beta_A A_s} g_{t-s} I_s, \\
\Exp[Y_t | \bar{A}_t, \bar{I}_t, \bar{Y}_{t-1}] &= \alpha_t \sum_{s<t} \pi_{t-s} I_s.
\end{aligned}       
    \end{equation} 
    We call \cref{eq::bhatt,eq::bhatt_R} the multiplicative version of the semi-mechanistic model and \cref{eq::exp_semi} the exponential version.
    Notice that for fixed values of $\bar{A}_t$, $\alpha_t$, $g$ and $\pi$,
    this model is a set of linear equations
    and is an example of what is known in the causality world as a linear {\em structural equation model} (SEM),
    for which tricks exist to derive the $g$-formula in closed-form.


    Meaningful dynamics in this model requires some positive infections $I_t$ prior to time $t = 1$.
    %
    \cite{bhatt} assumed
    $I_{t}=0$ for $t \leq -T_0$ and $I_{t}=e^\mu$ for $t = -T_0+1, \dots, 0$, where $\mu$ is a parameter to be estimated and $T_0 = 6$. Let $\bar{I}_0 \equiv (I_{-T_0+1}, \dots, I_0)$ indicate those \emph{seeding values} in the infection process. (We still exclude those seeding values in the definition of $\bar{I}_t = (I_1, \dots, I_t)$.)
    
    For the exponential model define
    $$
    \Lambda^e =
    \left(
    \begin{array}{cccc}
    0                             & 0                             & \cdots & 0\\
    g_1 e^{\beta_0 + \beta_A a_1} & 0                             & \cdots & 0\\
    g_2 e^{\beta_0 + \beta_A a_1} & g_1 e^{\beta_0 + \beta_A a_2} & \cdots & 0\\
    \vdots                        & \vdots                   & \vdots & \vdots\\
    g_{t-1} e^{\beta_0 + \beta_A a_1} & g_{t-2} e^{\beta_0 + \beta_A a_2} & \cdots  & 0
    \end{array}
    \right),
    \quad
    \Lambda^e_0 =
    \left(
    \begin{array}{ccc}
    g_{T_0} e^{\beta_0 + \beta_A a_{-T_0+1}}     & \cdots & g_1 e^{\beta_0 + \beta_A a_0}\\
    g_{T_0+1} e^{\beta_0 + \beta_A a_{-T_0+1}}   & \cdots & g_2 e^{\beta_0 + \beta_A a_0}\\
    g_{T_0+2} e^{\beta_0 + \beta_A a_{-T_0+1}}   & \cdots & g_3 e^{\beta_0 + \beta_A a_0}\\
    \vdots                                & \vdots & \vdots\\
    g_{T_0+t-1} e^{\beta_0 + \beta_A a_{-T_0+1}} & \cdots & g_t e^{\beta_0 + \beta_A a_0}
    \end{array}
    \right),
    $$
    and for the multiplicative model
    define 
    $$
    \Lambda^m =
    \left(
    \begin{array}{cccc}
    0                   & 0                       & \cdots       & 0\\
    g_1 R(\overline a_2, \beta) & 0                       & \cdots       & 0\\
    g_2 R(\overline a_3,\beta) & g_1 R(\overline a_3, \beta)   & \cdots       & 0\\
    \vdots              & \vdots                  & \vdots  & \vdots\\
    g_{t-1} R(\overline a_t, \beta) & g_{t-2} R(\overline a_t, \beta)& \cdots  & 0
    \end{array}
    \right),
    \quad
    \Lambda^m_0 =
    \left(
    \begin{array}{ccc}
    g_{T_0} R(\bar{a}_1, \beta)   & \cdots & g_1 R(\bar{a}_1,\beta)\\
    g_{T_0+1} R(\bar{a}_2, \beta)   & \cdots & g_2 R(\bar{a}_2,\beta)\\
    g_{T_0+2} R(\bar{a}_3, \beta)   & \cdots & g_3 R(\bar{a}_3,\beta)\\
    \vdots                        & \vdots & \vdots\\
    g_{T_0+t-1} R(\bar{a}_t, \beta) & \cdots & g_t R(\bar{a}_t, \beta)
    \end{array}
    \right).    
    $$
    Finally, define
    $$
    \Pi =
    \left(
    \begin{array}{cccc}
    0                   & 0                       & \cdots       & 0\\
    \pi_1 \alpha_2 & 0                       & \cdots       & 0\\
    \pi_2 \alpha_3 & \pi_1 \alpha_3   & \cdots       & 0\\
    \vdots              & \vdots                  & \vdots  & \vdots\\
    \pi_{t-1} \alpha_t & \pi_{t-2} \alpha_t & \cdots  & 0
    \end{array}
    \right),
    \quad
    \Pi_0 =
    \left(
    \begin{array}{ccc}
    \pi_{T_0} \alpha_1   & \cdots & \pi_1 \alpha_1\\
    \pi_{T_0+1} \alpha_2   & \cdots & \pi_2 \alpha_2\\
    \pi_{T_0+2} \alpha_3   & \cdots & \pi_3 \alpha_3\\
    \vdots                        & \vdots & \vdots\\
    \pi_{T_0+t-1} \alpha_t & \cdots & \pi_t \alpha_t
    \end{array}
    \right).    
    $$
    Then the marginal structural model $\psi(\bar{a}_t; \theta)$ is given in a closed form as follows.

    \begin{theorem}
        For the exponential model,
        \begin{equation}\label{eq::gform1}
            \E[I_t (\overline{a}_t)] = 
            [(id - \Lambda^e)^{-1} \Lambda^e_0 \bar{I}_0]_t
        \end{equation}
        and
        \begin{equation}\label{eq::gform2}
            \psi(\overline{a}_t;\theta) \equiv \E[Y_t(\overline{a}_t)] =
            [\{\Pi(id - \Lambda^e)^{-1} \Lambda^e_0 + \Pi_0\} \bar{I}_0]_t,
        \end{equation}
        where the subscript $t$ represents the $t$-th element of the outcome vector.
        For the multiplicative model,
        the expressions are the same except that 
        $\Lambda^m$ and $\Lambda^m_0$ replace $\Lambda^e$ and $\Lambda^e_0$. 
    \end{theorem}
    
    \begin{proof}
        Consider the intervened graph in \cref{fig::intervene} with
        $\overline{A}_t$ set to 
        $\overline{a}_t$.
        For this graph, 
        we have
        $$
        \overline{I}_t(\overline{a}_t) = \Lambda^e \overline{I}_t(\overline{a}_t) + \Lambda^e_0 \bar{I}_0 + \epsilon,
        $$
        for the exponential model (\cref{eq::exp_semi}),
        which,
        as mentioned above,
        is a linear structural equation model.
        Now
        $$
        \overline{I}_t(\overline{a}_t) = (id -\Lambda^e)^{-1} \Lambda^e_0 \bar{I}_0 + (id -\Lambda^e)^{-1} \epsilon
        $$
        and hence, the last element of this vector is
        $$
        \E[I_t (\overline{a}_t)] = [(id -\Lambda^e)^{-1} \Lambda^e_0 \bar{I}_0]_t.
        $$
        Subsequently,
        $$
        \Exp[\overline{Y}_t(\overline{a}_t)] = \Pi \Exp[\overline{I}_t(\overline{a}_t)] + \Pi_0 \bar{I}_0
        = [\{\Pi(id -\Lambda^e)^{-1} \Lambda^e_0 + \Pi_0\} \bar{I}_0]_t.
        $$

        The proof proceeds similarly for the multiplicative model (\cref{eq::bhatt}), 
        but with $\Lambda^m$ and $\Lambda^m_0$ in place of $\Lambda^e$ and $\Lambda^e_0$.
    \end{proof}   
\end{example}

Most often, we cannot solve the estimating equations (\cref{eq::ee}) in closed form. 
In that case, we apply Newton's method, 
which requires the computation of the first derivative on the left-hand side of the equation with respect to the parameter of interest. 
Specifically, since $\beta$ is the key parameter in both the semi-mechanistic model (\cref{eq::bhatt}) and the SEIR model (\cref{eq::seir}), 
we provide detailed calculations for the derivative with respect to $\beta$.

\begin{example}[Semi-mechanistic Model]
\label{ex::bhatt_derivative}
    In \cref{ex::bhatt}, we derived the closed-form expression for the 
    causal mean $\psi(\bar{a}_t; \theta)$ in a multiplicative semi-mechanistic model as:
    \begin{equation*}
        \psi(\bar{a}_t; \theta) =
        [\{\Pi(id - \Lambda^m)^{-1} \Lambda^m_0 + \Pi_0\} \bar{I}_0]_t.
    \end{equation*}
    For each component of $\beta$ ($\beta_0$ and $\beta_A$), the first derivative of $\Lambda^m$ with respect to $\beta_i$ is given by:
    \begin{equation*}
        \frac{\partial}{\partial \beta_i} \Lambda^m(t,s) = g_{t-s} \frac{\partial}{\partial \beta_i} R(\bar{a}_t, \beta) \mathbf{1}\{t > s\},
    \end{equation*}
    where $\Lambda^m(t,s)$ is parameterized by the rate function $R(\bar{a}_t, \beta)$, and $\mathbf{1}\{t > s\}$ is an indicator function. Similarly, the derivative of $\Lambda^m_0$ with respect to $\beta_i$ follows the same structure.
    
    Using the identity from matrix calculus, $\frac{\partial U^{-1}}{\partial x} = - U^{-1} \frac{\partial U}{\partial x} U^{-1}$, we can express the derivative of $\psi(\bar{a}_t; \theta)$ with respect to $\beta_i$ as:
    \begin{equation*}
    \begin{aligned}
        \frac{\partial}{\partial \beta_i} \psi(\bar{a}_t; \theta) 
        & = \left[ \left\{ 
            \Pi(id - \Lambda^m)^{-1} \frac{\partial \Lambda^m}{\partial \beta_i} (id - \Lambda^m)^{-1} \Lambda^m_0 + \Pi(id - \Lambda^m)^{-1} \frac{\partial \Lambda^m_0}{\partial \beta_i}
        \right\} \bar{I}_0 \right]_t \\
        & = \left[ \Pi(id - \Lambda^m)^{-1} \left\{ 
             \frac{\partial \Lambda^m}{\partial \beta_i} \Exp[\bar{I}_t(\bar{a}_t)] + \frac{\partial \Lambda^m_0}{\partial \beta_i} \bar{I}_0 
        \right\} \right]_t.
    \end{aligned}
    \end{equation*}
    The derivative for the exponential model is given similarly.
\end{example}

\begin{example}[SEIR Model]
\label{ex::seir_derivative}
The SEIR model in \cref{eq::seir} does not admit a closed-form expression for the 
causal effect $\psi(\bar{a}_t; \theta)$. 
In \cref{eq::monte_carlo}, we proposed estimating this quantity through 
Monte Carlo approximation. 
Here, we describe how the derivative of this approximation can also be computed using Monte Carlo samples.
First, applying the law of total probability,
\begin{equation*}
\begin{aligned}
\Exp_\theta[Y_t(\bar{a}_t)] &=
\int \cdots \int \Exp[Y_t | \bar{B}_{t-1}, \bar{C}_{t-1}, \bar{Y}_{t-1}] 
\prod_{s=1}^{t-1} dP_{Y_s}(Y_s | \bar{B}_{s-1}, \bar{C}_{s-1}, \bar{Y}_{s-1}) \\
& \quad \times dP_{C_s}(C_s | \bar{B}_{s-1}, \bar{C}_{s-1}, \bar{Y}_{s-1}) 
\times dP_{B_s|\beta}(B_s | \bar{B}_{s-1}, \bar{C}_{s-1}, \bar{Y}_{s-1}),
\end{aligned}
\end{equation*}
where $P_{Y_s}$, $P_{C_s}$, and $P_{B_s}$ are binomial distributions. 
(Note that $\overline{a}_s$ is implicit in some expressions on the right hand side.)
The ``number of trials" parameters for these variables depend on the 
conditioning terms $\bar{B}_{s-1}$, $\bar{C}_{s-1}$, and $\bar{Y}_{s-1}$, 
with success probabilities denoted by $p_Y$, $p_C$, and $p_B$, respectively. 
Importantly, only $p_{B_s|\beta}$, and consequently $P_{B_s|\beta}$, are parametrized by 
$\beta$ through $\eta(\bar{a}_s; \beta)$.
Now, suppose we have Monte Carlo samples 
$\{(\bar{B}_T^{(k)}, \bar{C}_T^{(k)}, \bar{Y}_T^{(k)}): k = 1, \dots, N\}$ 
drawn under a given $\beta$. For any alternative parameter $\beta'$, we approximate 
the mean using importance sampling as follows:
\begin{equation*}
\begin{aligned}
\psi(\bar{a}_t; \theta') 
&\approx \frac{1}{N} \sum_{k=1}^N \Exp[Y_t^{(k)} | \bar{B}_{t-1}^{(k)}, \bar{C}_{t-1}^{(k)}, \bar{Y}_{t-1}^{(k)}] 
\prod_{s=1}^{t-1} \frac{dP_{B_s|\beta'}}
{dP_{B_s|\beta}}(B_s | \bar{B}_{s-1}^{(k)}, \bar{C}_{s-1}^{(k)}, \bar{Y}_{s-1}^{(k)}),
\end{aligned}
\end{equation*}
where $\theta'$ is $\theta$ with component $\beta$ replaced by $\beta'$, 
and $\frac{dP_{B_s|\beta'}}{dP_{B_s|\beta}}$ is the Radon--Nikodym derivative. 
Note that when $\beta' = \beta$, the importance sampling estimator reduces to the Monte Carlo 
approximation for $\psi(\bar{a}_t; \theta)$ as given in \cref{eq::monte_carlo}.
    
Next, to compute the derivative of $\psi(\bar{a}_t; \theta)$ with respect to $\beta$, 
we recognize that the derivative of the Radon--Nikodym derivative $
\frac{dP_{B_s|\beta'}}{dP_{B_s|\beta}}$ at $\beta' = \beta$ is 
the gradient of the log-likelihood:
$\nabla_\beta \log\{f_{B_s|\beta}(B_s | \bar{B}_{s-1}^{(k)}, \bar{C}_{s-1}^{(k)}, \bar{Y}_{s-1}^{(k)})\}$. 
By applying the chain rule, we obtain:
\begin{equation*}
\begin{aligned}
\nabla_\beta \psi(\bar{a}_t; \theta) 
&\approx \frac{1}{N} \sum_{k=1}^N 
\Exp[Y_t^{(k)} | \bar{B}_{t-1}^{(k)}, \bar{C}_{t-1}^{(k)}, \bar{Y}_{t-1}^{(k)}] 
\sum_{s=1}^{t-1} \nabla_\beta 
\log\{f_{B_s|\beta}(B_s | \bar{B}_{s-1}^{(k)}, \bar{C}_{s-1}^{(k)}, \bar{Y}_{s-1}^{(k)})\}.
\end{aligned}
\end{equation*}
\end{example}

\section{Method 4 - Causal Preserving Data Generating Models}
\label{section::full}

Method 3 requires two models: an MSM
for $Y_t(\overline{a}_t)$
and a model for the propensity score. It does not require a 
model for the joint distribution
$p(\overline{x}_t,\overline{a}_t,\overline{y}_t)$ of the observables.
Method 4 consists of constructing a model
for
$p(\overline{x}_t,\overline{a}_t,\overline{y}_t)$
that preserves the specified epidemic MSM
for $Y_t(\overline{a}_t)$.
Then we can use
maximum likelihood to estimate all the parameters
of the model, including the parameters of the causal model.
We use the approach developed in
\cite{evans2024}
based on
copulas, which they call a 
{\em frugal parameterization}.
This requires more modeling than the
estimating equation approach
but it also avoids the null paradox
and has the advantage that one avoids dividing by the propensity score
in \cref{eq::pp0},
which can lead to unstable inference.
We will assume throughout this section
that all the variables are
continuous.

First, we recall some basics about copulas
\citep{joe2014}.
A copula 
$C(u_1,\ldots, u_d)$
is a joint distribution
on $[0,1]^d$ with uniform marginals.
We let
$c(u_1,\ldots, u_d)$
denote the corresponding density function.
A key fact is that
any joint density
$p(x_1,\ldots, x_d)$
for random variables
$X_1,\ldots, X_d$ can be written as
\begin{equation}\label{eq::cop1}
p(x_1,\ldots, x_d) = c(F_1(x_1),\ldots, F_d(x_d))
\prod_j p_j(x_j) 
\end{equation}
for some copula $c$,
where $p_j$ is the marginal density of $X_j$ and
$F_j$ is the corresponding cdf.
Thus, copulas provide a way to paste together
a set of marginal distributions to form a joint distribution.

One can consider parametric families of copulas
$c(u;\theta)$.
For example, the Gaussian copula has density
$$
c(u) = |\theta|^{-1/2} 
\exp\left(-\frac{1}{2} \Phi^{-1}(u)^T (\theta-I) \Phi^{-1}(u)\right)
$$
where $\theta$ denotes a correlation matrix,
$\Phi(u) = (\Phi(u_1),\ldots,\Phi(u_d))$
and $\Phi$ is the standard Normal cdf.
The process of constructing parametric families of copulas
has a rich literature.



To see how this helps build causal models, first consider observations at a single time point: a vector of confounders $X \in \mathbb{R}^d$, an intervention $A \in \mathbb{R}$, and an outcome $Y \in \mathbb{R}$. Let $p_a(y)$ be the density of a given marginal structural model for the counterfactual $Y(a)$, and let $F_a(y) = \int^y_{-\infty} p_a(s) ds$ denote the cdf.
We aim to construct a joint density $p(x,a,y)$ for the observed variables $(X,A,Y)$ that is consistent with the given counterfactual distribution $p_a(y)$. Consistency here means that applying the $g$-formula, i.e., $\int p(y|x,a)p(x) dx$, to the joint density $p(x,a,y)$ recovers the original counterfactual density $p_a(y)$.
Under Conditions (C1), (C2) and (C3), we can construct such joint distributions by
\begin{equation*}
    p(x,a,y) = p_a(x,a,y) = p_a(y) p(a)  p_j(x_1) \cdots p_d(x_d)
    c(F_a(y), G_1(x_1), \ldots, G_d(x_d), Q(a)).
\end{equation*}
using \cref{eq::cop1}, where
$p_a(x,a,y)$ denotes the joint density of
$(X,A,Y(a))$, and $G_j$ and $Q$ are the cdfs of $X_j$ and $A$, respectively.
We can add parameters to these distributions to define a parametric family
\begin{equation*}
\begin{aligned}
    p(x,a,y;\beta,\gamma,\theta) 
    & = p_a(y;\beta) p(a;\delta)\prod_j p_j(x_j;\gamma_j) \\
    & \quad \times c(F_a(y;\beta),G_1(x_1;\gamma_1),\ldots, G_d(x_d;\gamma_d),Q(a;\delta); \theta).
\end{aligned}
\end{equation*}
Because $\beta$ and $\theta$ parametrize the causation by $A$ and other indirect correlation between $A$ and $Y$ separately, these models are variation independent and avoid the $g$-null paradox \citep{evans2024}.

Turning to the time varying case,
a similar construction can be used
but is much more involved.
The recent paper by
\cite{lin2025} shows how to use
a class of copulas known as
pair copulas to parameterize the joint
distribution.
The details are fairly involved
and we refer the reader to
\cite{lin2025} for details.

The estimating equation approach (Method 3) requires two models:
the epidemic MSM
and a model for the propensity score.
The fully specified, frugal approach (Method 4) requires, in addition,
a model for $X$ and a copula.
A full exploration of how to
construct these models
will be quite complicated and we leave this
for future work.
The advantage of Method 3 is thus that it requires
less modeling.
The advantage of Method 4
is that we never need to divide by the propensity score
which can lead to instability.
Also, a fully specified
joint distribution (Method 4)
may be useful for purposes of interpretability and
model checking.

\section{Examples} \label{sec::example}

Our first two examples use simulated data to illustrate 
that phantom variables can induce bias in ML parameter estimates, 
and that using estimating equations yields unbiased estimates. 
The last example is an analysis of the effect of 
a mobility measure -- the proportion of full-time work --
on COVID-19 deaths in 30 
US states at the start of the pandemic.
The code vignettes used to generate the results 
are provided in \url{github.com/HeejongBong/causepid}.

\subsection{Semi-mechanistic model simulated data} \label{sec::simulation_semi}

A time series consistent with the DAG in \cref{fig::dagsemi0} is 
simulated from the semi-mechanistic model in \cref{eq::bhatt} as follows. 
\begin{itemize}
\item
Seed the infection time series by setting \(I_t = 0\) for \(t \leq -40\) and \(I_t = e^{\mu}\),
\(\mu = \log(100)\), for \(t = -39, \dots, 0\). 
The $t=-40$ time cutoff point corresponds to the 
main support of the generating distribution $g$ used 
in \cite{bhatt,bong2024} and shown in \cref{fig::g}.
(More realistic infection data could be simulated from the 
data generating process described below starting at $t=-39$, 
but then it is difficult to obtain ML estimates. Details 
are in \cite{bong2024}.)
\item
Seed the confounder, intervention and outcome at $t = 0$ with $X_0 = A_0= 0$ and $Y_0 = I_{-1}$.
  \item  
Simulate phantom variables \(U_t\) from a Gaussian random 
process with mean zero and covariance kernel \(\Sigma(t, s) = \phi^{\abs{t-s}}\), with \(\phi = 0.95\).
\end{itemize}

Then for \(t = 1, \dots, 120\),
\begin{enumerate}
    \item 
    sample confounders \(X_t\) from a Gaussian distribution with mean \( \xi_U U_t + \xi_X X_{t-1} + \xi_A A_{t-1} + \xi_Y Y_{t-1}\), with \(( \xi_U, \xi_X, \xi_A, \xi_Y) = ( 0.2, 0, 1, 0)\) and variance \(\sigma^2 = 0.09\);
    
    \item 
    generate binary interventions \(A_t\) from a Bernoulli distribution with probability 
    \[
        \P(A_t = 1 \mid \bar{X}_t, \bar{A}_{t-1}, \bar{Y}_{t-1}) = \frac{e^{\gamma_1 + \gamma_X X_t + \gamma_A A_{t-1} + \gamma_Y Y_{t-1}}}{1 + e^{\gamma_1 + \gamma_X X_t + \gamma_A A_{t-1} + \gamma_Y Y_{t-1}}}, 
    \]
    where \((\gamma_1, \gamma_X, \gamma_A, \gamma_Y) = (-2.5, 0, 4, 0.001)\); 
    
    \item 
    simulate an infection process \(I_t\) from a negative binomial distribution with ``number of successes'' parameter \(\nu = 10\) and 
    mean parameter specified in \cref{eq::bhatt}, where $g$ is the generating distribution in \cref{fig::g}, 
    the reproduction number is 
    \begin{equation} \label{eq::repro}
    R(\bar{A}_t, \beta) = \frac{K}{1 + \exp(\beta_0 + \beta_U U_t + \beta_X X_t + \beta_A A_t)},
    \end{equation}
    \(K = 6.5\) is the maximum transmission rate, {and $(\beta_U, \beta_X) = (0.3,0)$.
    We will vary the values of $\beta_0$ and $\beta_A$, as described below.}

    {Notice that we set $\gamma_X=\beta_X=0$, meaning that there is no unobserved confounding, to highlight that phantom bias is 
    different than unobserved confounder bias.}
    
    \item 
    Finally, simulate an observed time series $Y_t$ -- e.g. cases or deaths -- according to \cref{eq::bhatt} with \(\alpha_t = 1\) and \(\pi_t = \mathbf{1}\{t=1\}\), for simplicity, so that \(Y_t = \mathbb{E}[Y_t \mid \bar{I}_{t-1}, \bar{Y}_{t-1}, \bar{A}_t] \equiv I_{t-1}\) for all \(t\). 
\end{enumerate}

\begin{figure}[t!]
    \centering
    \includegraphics[height=0.315\textwidth]{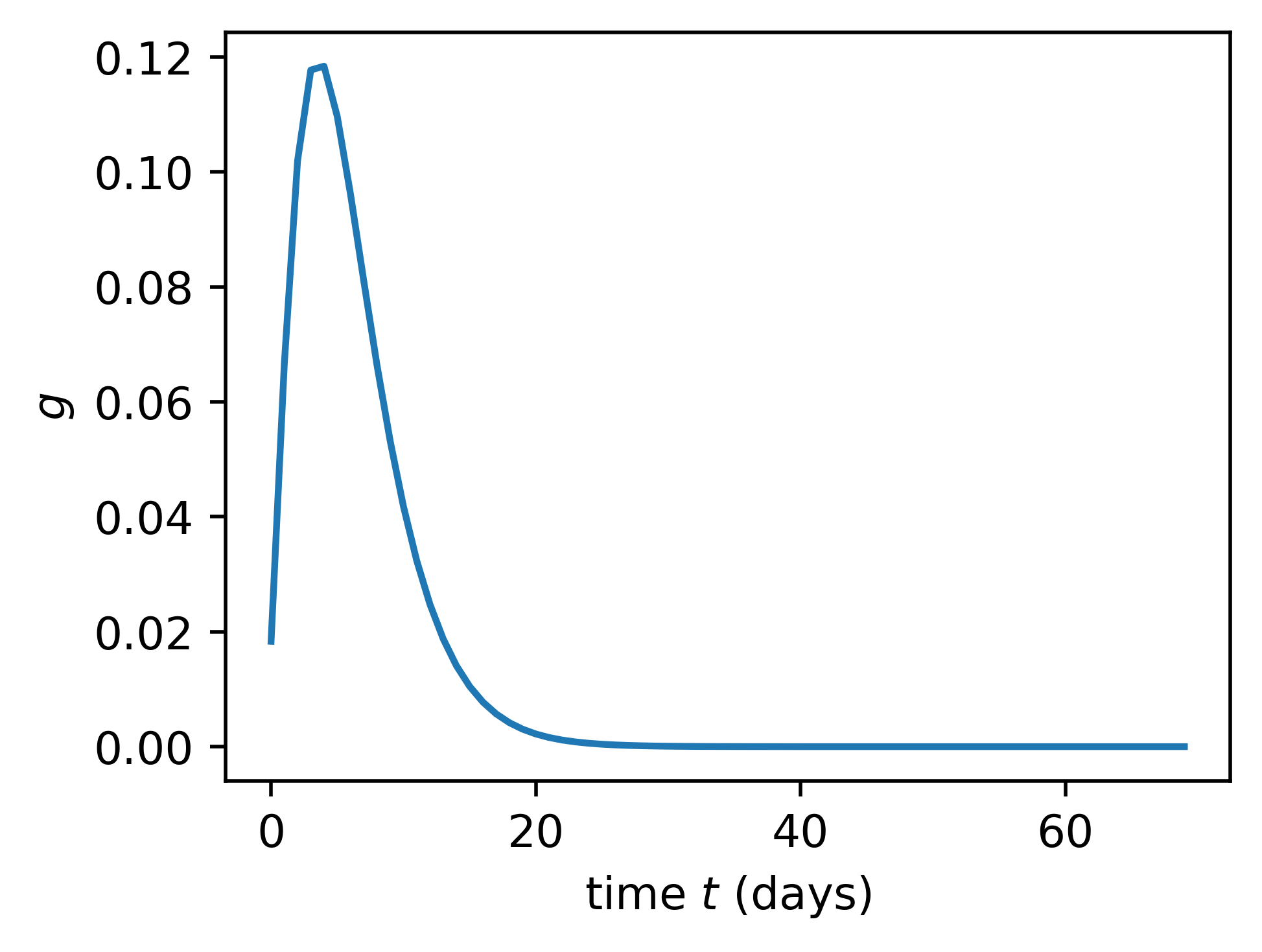}
    \vspace{-0.1in}
    \caption{\sl Generating distribution $g$ from \citet{bhatt}.}
    \label{fig::g}
\end{figure}



The simulation was performed for $21$ linearly spaced values
of $\beta_A$ in \([-1, 0]\), and 
for each value of $\beta_A$ we let $\beta_0$ depend on $\beta_A$ according to
$\beta_0 = - \log(5.5) + 0.5 - \beta_A/2$, to prevent
$R(\bar{A}_t, \beta)$ from getting too small or too large, so that the 
simulated epidemic curves do not explode or plunge to zero. Our parameter 
choices mostly produced epidemic curves with shapes we typically observe in practice, that is
rise, plateau and then slowly decrease.
For each $\beta_A$, we simulated 200 time series $Y_t$, \(t = 1, \dots, 120\), 
and for each time series, {we obtained the ML estimates of $(\beta_0, \beta_A, \beta_X)$ and
the seeding parameter $\mu$ (this is Method 1)
using 
the package \texttt{freqepid} \citep{bong2024}, assuming the same DGM we used to simulate the data
but for one detail: we assumed reproduction number
$R(\bar{A}_t, \beta) = \frac{K}{1 + \exp(\beta_0 + \beta_X X_t + \beta_A A_t)}$ (\cref{eq::Rconfound})
instead of \cref{eq::repro}, which was used to simulate the data, since the 
$U_t$ are not observed. All other parameters were assumed to be known.}

\cref{fig::conf_int_ML_bhatt} shows the averages of the 200 MLEs of $\beta_A$ plotted against the true $\beta_A$, along
with 95\% confidence intervals.
There is phantom bias.
There was no bias when we reproduced the simulation with \(\beta_U = 0 \) in \cref{eq::repro}, confirming that
what
we see in \cref{fig::conf_int_ML_bhatt} 
is due to phantom variables; see \cref{fig::conf_int_non_phantom_betas}.

\begin{figure}[t!]
    \centering
    \setlength{\labelsep}{-2.5pt}
    \begin{minipage}{0.45\textwidth}
        \sidesubfloat[]{\label{fig::conf_int_ML_bhatt} 
        \includegraphics[width=0.95\textwidth]{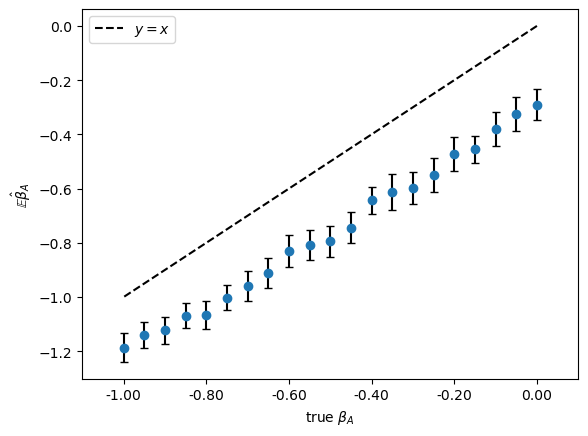}}
    \end{minipage}
    \hspace{0.2in}
    \begin{minipage}{0.45\textwidth}
        \sidesubfloat[]{\label{fig::conf_int_EEnc_bhatt} 
        \includegraphics[width=0.95\textwidth]{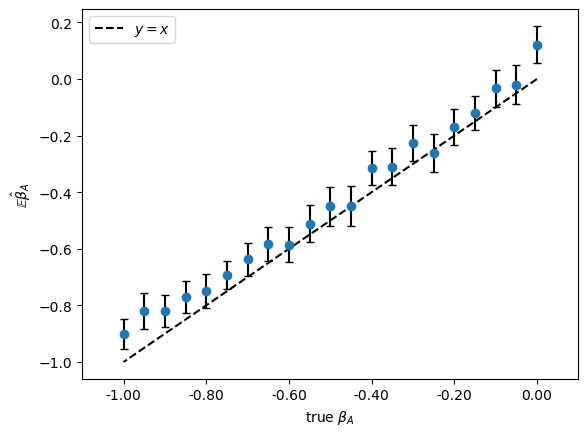}}
    \end{minipage}
    \vspace{-0.1in}
    \caption{\sl {\bf Causal parameter estimation for semi-mechanistic epidemic model data.} (a) Method 1: ML estimate and (b) Method 3: 
    estimating equations estimate of \(\beta_A\) averaged across 200 repeat simulations (blue dots) with 95\%-confidence intervals (error bars), for a range of true $\beta_A$ values. There is phantom bias in (a) but not in (b). 
    Small biases persist in (b), as evidenced by the confidence intervals excluding the true value 
of $\beta_A$ more often than they should. This is due to model misspecification, which we 
discuss in \cref{sec::misspec}.
    \label{fig::conf_int_bhatt}
    }
\end{figure}

For each simulated time series we estimated $(\beta_0, \beta_A)$ per Method 3,
by solving the estimating equations in \cref{eq::ee},
assuming the MSM in \cref{eq::bhatt} with reproduction number in \cref{eq::bhatt_R},
that is $R(\bar{A}_t, \beta) = \frac{K}{1 + \exp(\beta_0 + \beta_A A_t)}$. 
(Note that \cref{eq::bhatt_R} depends only on $A_t$.)
We proceeded numerically via Newton's method initialized at the MLEs, 
and causal effect
$\psi(\overline{a}_t, \theta)$ and derivative $\nabla_\beta \psi(\overline{a}_t, \theta)$ obtained 
following \cref{ex::bhatt,ex::bhatt_derivative}.
{In~\cref{eq::ee}, we took $h_t(\bar A_t) = (1, A_t)$ and we
modeled the denominator of the propensity score $W_t$ (\cref{eq::pp0})
using a logistic regression with linear AR(1) logit link, 
$$
\mathrm{logit}(\pi_t)  =  \chi_0 + \chi_1 A_{t-1} + \gamma_0 X_t + \lambda_1 Y_{t-1},
$$
as described in \cref{sec::propen},
and estimated the parameters by maximum likelihood.
We proceeded similarly with a logistic AR(1) model for the numerator.}

\cref{fig::conf_int_EEnc_bhatt} shows the averages over 
the 200 simulations of the estimating equation estimates of 
$\beta_A$ for each true value of $\beta_A$, along 95\% confidence intervals. 
We no longer have phantom bias.
{Notice that small biases still persist, as evidenced by the confidence intervals excluding the true value 
of $\beta_A$ more often than they should. This is due to model misspecification, which we 
discuss in \cref{sec::misspec}.
}

\subsection{SEIR model simulated data} \label{sec::simulation_seir}

We further illustrate phantom bias using the SEIR model in \cref{ex::seir}.
One time series is generated as follows.

\begin{itemize}
    \item Set $I_0 = 100$, $E_0 = 0$ and $S_0 = N - I_0 - E_0$, with total population $N = 100,000$.

    \item Simulate phantom variables \(U_t\) from a Gaussian random process with mean 
    zero and covariance kernel \(\Sigma(t, s) = \phi^{\abs{t-s}}\), with \(\phi = 0.95\). 
\end{itemize}

Then for \(t = 1, \dots, 120\), 
\begin{enumerate}
    \item
    sample confounders \(X_t\) from a Gaussian distribution with mean \(\xi_U U_t + \xi_X X_{t-1} + \xi_A A_{t-1} + \xi_Y Y_{t-1}\), 
    where \((\xi_U, \xi_X, \xi_A, \xi_Y) = (0.5, 0, 1, 0)\) and variance \(\sigma^2 = 0.09\).
    \item Generate binary interventions \(A_t\) from a Bernoulli distribution with
    \[
        \P(A_t = 1 \mid \bar{X}_t, \bar{A}_{t-1}, \bar{Y}_{t-1}) = \frac{e^{\gamma_1 + \gamma_X X_t + \gamma_A A_{t-1} + \gamma_Y Y_{t-1}}}{1 + e^{\gamma_1 + \gamma_X X_t + \gamma_A A_{t-1} + \gamma_Y Y_{t-1}}},
    \]
    where 
    \((\gamma_1, \gamma_X, \gamma_A, \gamma_Y) = (-2.5, 0, 4, 100/N)\);
    \item 
    simulate an exposure process \(B_t\) from a binomial distribution with number of trials $S_{t-1}$ and success probability 
    \begin{equation}\label{eq::pB}
        p_{B}(\bar A_t, \beta) = 1 - \exp(- \eta(\bar A_t, \beta)  I_{t-1} / N),
    \end{equation}
    with \(\eta(\bar A_t, \beta) = \exp(- \beta_0 - \beta_U U_t - \beta_X X_t - \beta_A A_T - \beta_Y Y_{t-1})\),
    and set $S_t = S_{t-1} - B_t$.
    
    We considered $21$ linearly spaced values
    of $\beta_A$ in \([-1, 0]\), and 
    for each value, we set 
    \((\beta_0, \beta_U, \beta_X) = (1 - \beta_A/2, 0.3, 0)\) to keep 
    $\eta_t$ in the same ballpark for all values of $\beta_A$;

    \item simulate an infection process $C_t$ from a binomial distribution with number of trials $E_{t-1}$ and success probability $p_C = 0.2$,
    and set $E_t = E_{t-1} + B_t - C_t$;

    \item simulate an removal process $D_t$ from a binomial distribution with number of trials $I_{t-1}$ 
    and success probability $p_D = 0.2$, and set $I_t = I_{t-1} + C_t - D_t$;
    
    \item and finally, simulate an observed time series $Y_t \equiv D_t$ for all $t$.
\end{enumerate}

For each value of $\beta_A$, we simulated 200 time series $Y_t$ from the SEIR model, and
for each time series, we estimated $(\beta_0, \beta_X, \beta_A)$ using Method 1, assuming 
the SEIR model as the DGM. 
{The other parameters were assumed to be known.}
Because we don't have a developed method for obtaining the MLE in this setting, 
we used a regressive approach. 
Specifically, we fitted the binomial regression model:
\begin{equation*}
    B_t \sim \mathrm{Binomial} \left(S_{t-1}, \exp(- \beta_0 - \beta_X X_t - \beta_A A_t + \log(I_{t-1}/N))\right),
\end{equation*}
where $\log(I_{t-1}/N)$ is an offset. 
This approximates the generation of the exposure process because \cref{eq::pB}~$ \approx \eta(\bar A_t, \beta)  I_{t-1} / N$,
since $N$ is much larger than the numerator.

\cref{fig::conf_int_ML_seir} shows the averages over 
the 200 simulations of the regressive estimates of $\beta_A$ for each true value of $\beta_A$, along 95\% confidence intervals. 
Phantom bias is evident.

Next we estimated $(\beta_0,\beta_A)$ by solving the estimating equations in \cref{eq::ee} 
using Method 3, 
assuming the SEIR model with $\eta(\bar A_t, \beta) = \exp(- \beta_0 - \beta_A A_t)$ as MSM,
{$h_t(\bar A_t) = (1, A_t)$, 
and modeling the numerator and denominator of the propensity score $W_t$ in \cref{eq::pp0}
using logistic regressions with linear AR(1) logit links, as in the previous example.}
Because
there is no closed-form expression for $\psi(\overline{a}_t, \theta)$ or its derivative 
$\frac{\partial}{\partial\beta} \psi(\overline{a}_t, \theta)$, we used
the Monte Carlo approximations in \cref{eq::monte_carlo} and \cref{ex::seir_derivative}, 
and solved the estimating equations using Newton's method initialized at the regressive estimates. 

\cref{fig::conf_int_EEnc_seir} shows the averages over 
the 200 simulations of the estimating equation estimates of 
$\beta_A$ for each true value of $\beta_A$, with 95\% confidence intervals. 
We no longer have phantom bias.

\begin{figure}[t!]
    \centering
    \setlength{\labelsep}{-2.5pt}
    \begin{minipage}{0.45\textwidth}
        \sidesubfloat[]{\label{fig::conf_int_ML_seir} 
        \includegraphics[width=0.95\textwidth]{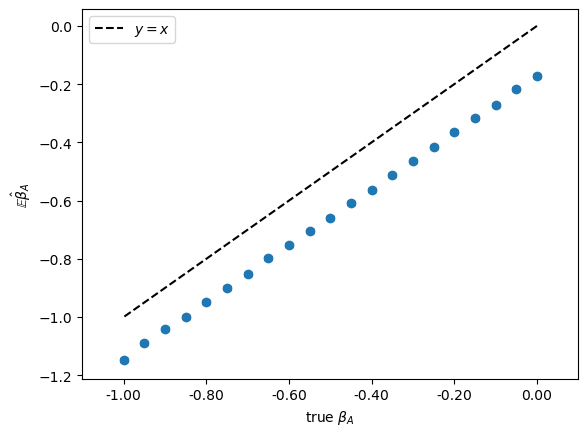}}
    \end{minipage}
    \hspace{0.2in}
    \begin{minipage}{0.45\textwidth}
        \sidesubfloat[]{\label{fig::conf_int_EEnc_seir} 
        \includegraphics[width=0.95\textwidth]{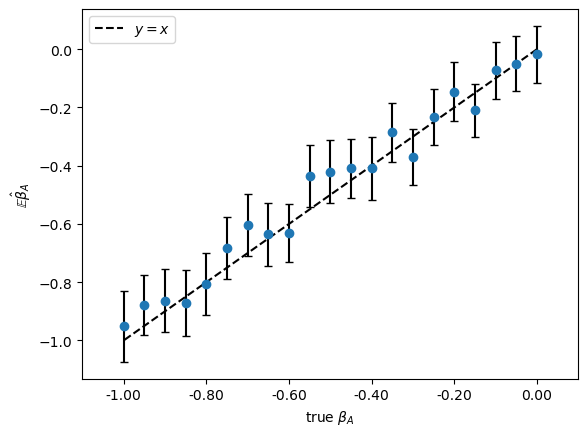}}
    \end{minipage}
    \vspace{-0.1in}
    \caption{\sl  {\bf Causal parameter estimation for SEIR epidemic model data.} (a) ML estimate (Method 1) 
    and (b) estimating equation estimate (Method 3) of \(\beta_A\) averaged across 200 repeat simulations (blue dots) with 95\%-confidence intervals, for a range of true $\beta_A$ values (the errors bars are too small to be seen in (a)). 
    There is phantom bias in (a) but not in (b).}
    \label{fig::conf_int_seir}
\end{figure}

\subsection{Effect of Mobility on COVID-19 Transmission} \label{sec::covid}
We analyzed the effect of a mobility measure on COVID-19 death data for U.S. states, 
using the dataset described in \citet{bong2024}. The data are sourced from the Delphi 
repository at Carnegie Mellon University (\url{delphi.cmu.edu}), and consist of 
daily observations from February 15 to August 1, 2020 (168 days). The dataset 
includes state-level records of COVID-19 deaths, denoted as $Y_t$, and a mobility 
measure $A_t$, “proportion of full-time work”, which represents the fraction of mobile 
devices that spent more than six hours at a location other than their home during daytime 
(using SafeGraph's \texttt{full\_time\_work\_prop}). We focused on the 30 states that 
reported more than 20 deaths on at least one day and truncated the time series 30 days 
prior to reaching a total of 10 accumulated deaths, following the procedure outlined in \citet{bhatt}. 
A preprocessing step was used to correct for the weekend effect, which shows fewer deaths 
reported on Saturdays and Sundays and, to compensate, more deaths reported on Mondays and Tuesdays
(see \citet{bong2024} for further details).

{We modeled these data using the augmented semi-mechanistic 
epidemic model in \cref{eq::bhatt,eq::bhatt_R} and seeding mechanism described
in \cref{sec::simulation_semi}. That is, we used the same DGM
as in \cref{sec::simulation_semi},
except that here, we do not have covariates $X_t$ and the treatment $A_t$ 
is continuous rather than binary.
\cref{sec::simulation_semi} also contains the 
details for obtaining the ML estimates of $(\mu, \beta_0, \beta_A)$ 
and estimating equation estimates 
of $(\beta_0, \beta_A)$, although, to estimate the propensity score,
we used AR(1) Gaussian regressions with identity linear links 
instead of logistic regressions,
since the treatment is continuous.}

\cref{fig::beta_EEnc_delphi} shows the estimates of $\beta_A$ 
for the 30 states.
The faint thick lines show the point estimates and 95\% confidence intervals
calculated separately for each state. These estimates can be improved by 
borrowing strength across states using a frequentist approach based on the 
robust empirical Bayes shrinkage method introduced by \citet{armstrong2022robust}, 
and extended to multivariate parameters by \citet{bong2024}. 
The dark thin estimates and intervals are the results of this procedure.

\cref{fig::beta_freqepid_delphi} shows the maximum likelihood estimates 
of $\beta_A$
from \citet{bong2024} based on the same model assumptions, and further
assuming $Y_t$ was negative binomial with mean in \cref{eq::bhatt}
to 
complete the DGM.

All estimates are positive and a few are not significantly different from zero, suggesting that 
reduced mobility lowered or had no significant effect on COVID-19 deaths.
However, there are substantial differences between ML and estimating equation estimates.
In 24 out of 30 states, the estimating equation estimates are lower than ML estimates. 
This result is statistically significant, assuming a binomial probability of 
$0.5$ for the two methods yielding smaller estimates equally across all 
states ($p < 0.001$). This suggests the possible presence of a phantom effect, 
leading to 
ML estimates overestimating the causal effect of mobility on COVID-19 deaths.

\begin{figure}[t!]
    \centering
    \setlength{\labelsep}{-2.5pt}
    \begin{minipage}{0.318\textwidth}
        \sidesubfloat[]{\label{fig::beta_EEnc_delphi} 
        \includegraphics[width=\textwidth]{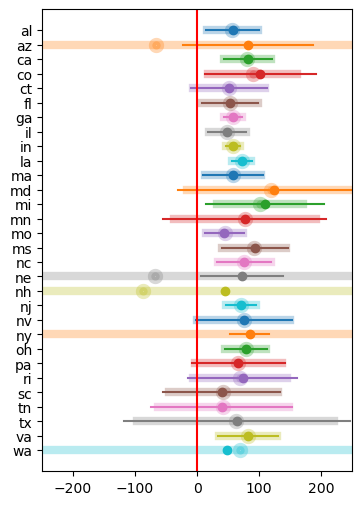}}
    \end{minipage}
    \hspace{0.5in}
    \begin{minipage}{0.318\textwidth}
        \sidesubfloat[]{\label{fig::beta_freqepid_delphi} 
        \includegraphics[width=\textwidth]{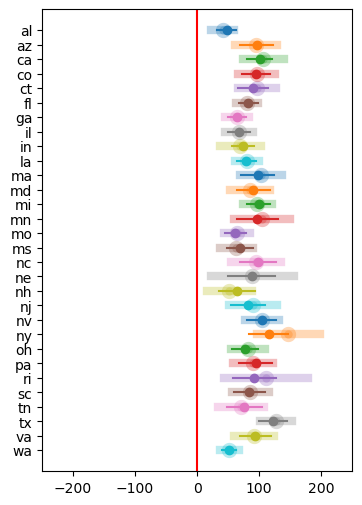}}
    \end{minipage}
    \vspace{-0.1in}
    \caption{\sl {\bf Effect of a mobility measure on COVID-19 death data for 30 U.S. state, 
    measured by $\beta_A$ in~\cref{eq::bhatt_R}.} Estimates and 95\% confidence intervals 
    using (a) the estimating equations in~\cref{eq::ee} (Method 3) and (b) maximum likelihood (Method 1). 
    The faint thick lines are the estimates and intervals before 
    shrinkage and the dark thin lines are the estimates after shrinkage.}
  \label{fig::beta_delphi}
\end{figure}

\section{Model Misspecification\label{sec::misspec}}

{To quote George E. P. Box, "All models are wrong."
Consider again the example in \cref{sec::simulation_semi}.
\cref{fig::conf_int_EEnc_bhatt} illustrated 
that obtaining the estimate for the causal 
parameter $\beta_A$ using estimating equations eliminated phantom bias.
However, the estimate is still a little biased. 
This is not unexpected since the MSM we used is 
not consistent with the model that generated the data, that is,
the MSM
is misspecified.
But the bias is much less than the bias of the 
MLE in \cref{fig::conf_int_ML_bhatt}, which suffers from
both misspecification bias (the likelihood we used is also misspecified
since it does not include the unobserved phantoms $\bar U_T$) and phantom bias.  
Using Methods 2 and 4 would allow us to illustrate phantom bias 
without model mispecification bias in a simulation study, but they are very onerous.
With Method 2, instead of assuming an MSM and deriving the
causal effect $\psi(\overline{a}_t;\theta)$ from it,
we would obtain
$\psi(\overline{a}_t;\theta)$
by applying the 
$g$-formula in \cref{eq::g3} to the joint distribution of 
$(\overline{U}_T,\overline{X}_T,\overline{A}_T,\overline{Y}_T)$
we used to simulate the data. We would need to proceed by Monte Carlo as in 
\cref{eq::monte_carlo}, which would be computationally very costly.
With Method 4, we would construct a joint
distribution $P$ for $(\overline{U}_T,\overline{X}_T,\overline{A}_T,\overline{Y}_T)$ such
that, when the $g$-formula is applied to $P$, we would get back
the causal effect $\psi(\overline{a}_t;\theta)$
derived from our assumed MSM. This can be done as described in
\cref{section::full}, using the construction due to \cite{evans2024}.
However, doing so is quite complicated in this setting.  }

{
Notice that to proceed with Methods 2 and 4 as we just explained would require
that the 
phantoms $\overline{U}_T$ be modeled. This is feasible in a simulation study, 
but not with real data since $\overline{U}_T$
are not observed.
Therefore, in practice, all methods, including Methods 2 and 4, 
will suffer some degree of model misspecification.
However, only Methods 2, 3 and 4 eliminate phantom bias. Method 1 does not.
}

{More generally, with all statistical models,
there is always the danger of model misspecification.
But here, the role of
model misspecification
is quite different
for MSM's and DGM's.}
{
If a MSM is misspecified
then 
the estimating equations
at least give us
us an estimate of the projection of
$\E[Y_t(\overline{a}_t)]$
onto the causal model,
as discussed in
\cite{neugebauer2007, diaz}.
}
{
The misspecification
for DGM's is more insidious.
As we have seen when discussing Method 1,
essentially no sequentially parameterized DGM
contains any distribution that
has correlation between outcome and treatment
while having no causal effect.
This is a fundamental structural limitation of these models
unrelated to not being able to fit the data well.
}

\section{Conclusion\label{sec::conclusion}}

To assess the effect of interventions,
one can
add an intervention
variable to an epidemic model.
There are two interpretations of this augmented model:
it is a causal model for a counterfactual
or it is a data generating model for the
observed variables.
In the literature, these have often been 
treated interchangeably but,
in general, these are not the same.
How we estimate the parameters
depends on which interpretation we use.
We have discussed
four approaches.
See Table \ref{table::4methods}.
Here we summarize the advantages and disadvantages of each.

\begin{table}[t!]
{
\begin{tabular}{l|ll}
Method & Algorithms          & Specification\\ \hline
I      & mle                 & joint distribution for $(\overline{X}_t,\overline{A}_t,\overline{Y}_t)$\\
II     & estimating equations & joint distribution for $(\overline{X}_t,\overline{A}_t,\overline{Y}_t)$\\
III    & estimating equations & MSM for $\overline{Y}_t(\overline{a}_t)$, propensity score for $\overline{A}_t$ given the past\\
IV     & mle                 & MSM for $\overline{Y}_t(\overline{a}_t)$ and copula for other variables
\end{tabular}
}
\caption{{\sl Summary of methods. In addition to the above, all the methods require
the three causal assumptions (no interference, positivity and no unmeasured confounding).}}
\label{table::4methods}
\end{table}

Method 1:
Use the augmented epidemic model as a data generating model,
use maximum likelihood or a Bayesian approach to estimate its parameters,
and then estimate the causal effect.
This 
leads to inconsistent estimates
and the $g$-null paradox.
In particular, it leads to non-zero estimates of the causal effect
even when there is no causal affect.

Method 2:
Use the augmented epidemic model as a data generating model and
apply the $g$-formula to that DGM to 
extract the causal model.
Also use the DGM to derive the propensity score.
Then use estimating equations
to estimate the parameters.
This avoids the $g$-null paradox but the 
method is quite cumbersome and the causal model is not 
easily interpretable.
An additional problem is that 
we have to divide by the propensity score,
which can cause the parameter estimates to
have large variances if the propensity score
gets small.

Method 3: 
Use the augmented epidemic model as a marginal structural model.
Postulate a model for the propensity score.
Then use estimating equations
to estimate the model parameters.
This is the simplest approach, and because it treats the augmented epidemic model as a 
model of how the intervention affects the outcome, it
facilitates the interpretation of
treatment on outcome.
But like Method 2, we have to divide by the propensity score,
so parameter estimates can have large variances.

Method 4:
Use the augmented epidemic model as a marginal structural model.
Construct a model for the data generating process that is consistent 
with the marginal structural model, and then use maximum likelihood
to estimate the model parameters.
The advantages of this method are that it avoids the $g$-null paradox and
dividing by the
propensity score. 
The disadvantages are that it is complicated and
requires more modeling assumptions.
The method also requires more investigation.

{
Methods 2, 3 and 4 avoid phantom bias.
If we have a good understanding of the full data generating process, 
then building a DGM for Method 2 may offer a more principled 
and realistic generative perspective. But Method 2 does not allow 
one to specify a model
for the treatment effect. Rather, the treatment effect is implied by
the $g$-formula, will not typically be available in closed form
and will not be easily interpretable.
Method 3 does not need the full DGM but still requires a
model for the propensity score, so at least part of what we understand 
about the data generating process can be incorporated there.
But unlike Method 2, Method 3 allows us to specify an interpretable 
causal model for the treatment effect. In Method 4, as in Method 3, we explicitly state the model for 
the effect of treatment, and we then build a full DGM around that 
specification. 
So Method 4 is interpretable and arguably more principled than Method 3.
However, Method 4 is quite new
and there is not yet a lot of experience that would
inform the practical aspects of implementing it
especially in the case of a single time series.
In particular, constructing and estimating the copula model
is non-trivial. So Method 4 seems promising from the 
aspect of interpretability and in terms of avoiding phantom bias
but more work is needed to see how practical it is.}

Whichever approach one uses,
it is important to distinguish
causal models and
data generating models.
And, when estimating 
parameters, it is important to account
for confounding.
No matter how many confounders we include
in an analysis, there is always the danger
that there are important unobserved confounders.
There are some methods for
dealing with unobserved confounding.
One of the oldest is to include
instrumental variables which are variables that
affect intervention but do not directly affect the outcome
\citep{greenland2000}.
More recently, there has been a surge of interest
in using negative controls, which are variables unaffected by intervention,
as a way to control for
unobserved confounding
\citep{tchetgen2024}.
We will report on these methods as applied to epidemic modeling
in future work.

{
Although completely nonparametric inference is not possible
for a single time series,
a referee suggested that using a partially nonparametric
data generating model might reduce phantom bias and misspecification bias
and might still be estimable.
This is a reasonable suggestion that deserves investigation.
}

{
Finally, let us mention that another interesting direction
would be to develop structural nested models (SNM)
\citep{vansteelandt2014structural, robins2000b}
for epidemics.
These models allow one to model
effect modification using
models for counterfactual contrasts that include covariates.
An example of a structural nested model in a simple setting is
$$
\E[Y(a)|X=x,A=a] - \E[Y(0)|X=x,A=a] = a(\psi_0 + \psi_1 x).
$$
The causal parameter
$\psi$ can be estimated using appropriate estimating equations.
These models also exist for time varying setting.
Building SNM's that somehow
incorporate epidemic models
is an open problem that deserves further investigation.
}

\bibliographystyle{apalike}
\bibliography{2_ref.bib}

\clearpage
\appendix
\counterwithin{figure}{section}
\counterwithin{equation}{section}

\section{Appendix}

\begin{figure}[h!]
    \centering
    \includegraphics[height=0.315\textwidth]{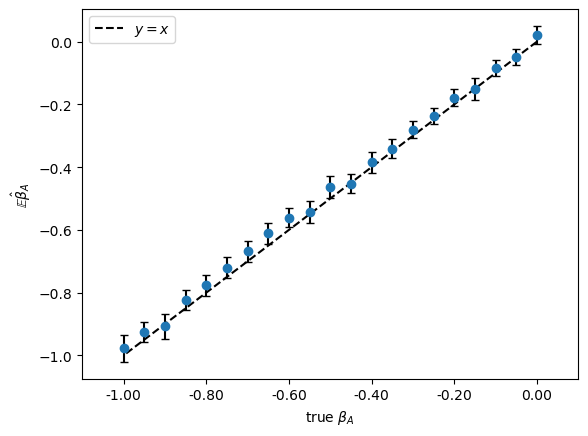}
    \caption{\textbf{Point estimates (blue dots) and 95\%-confidence intervals (error bars) of \(\beta_A\) from the ML estimates without phantom variables .} For each true \(\beta_A\) value, the point estimate and confidence interval were obtained from 200 i.i.d. estimates \(\hat{\beta}_A\). The ML estimates are unbiased when phantom variables are absent.}
    \label{fig::conf_int_non_phantom_betas}
\end{figure}

\end{document}